\documentclass[letterpaper,11pt]{article}
\usepackage[utf8]{inputenc}
\usepackage[margin=1in]{geometry}

\usepackage{amsmath}
\usepackage{amsthm}
\usepackage{amssymb}
\usepackage{thmtools}
\usepackage{thm-restate}
\usepackage{enumitem}
\usepackage{dsfont}
\usepackage{xspace}
\usepackage{comment}
\usepackage{xparse}
\usepackage{tikz}
\usepackage{fixme}
\usepackage[colorlinks]{hyperref}
\usepackage[capitalize]{cleveref}
\usepackage{subcaption}

\newcommand{\kbnote}[1]{{\color{blue}[#1]}}

\setlist[enumerate]{nosep, topsep=1ex}
\setlist[itemize]{nosep, topsep=1ex}
\setlist[description]{nosep,topsep=1ex}

\newtheorem{property}{Property}[]
\newtheorem{problem}{Problem}[]
\newtheorem{theorem}{Theorem}[]

\newtheorem{fact}{Fact}[]

\newtheorem{definition}{Definition}[]
\newtheorem{remark}{Remark}[]

\newtheorem{lemma}{Lemma}[]

\usepackage{algorithm}
\usepackage{algpseudocode}
\usepackage{algorithmicx}  
\usepackage{float}
\usepackage{microtype}
\usepackage{algorithm}
\usepackage{algpseudocode}
\usepackage{algorithmicx}   
\usepackage{floatflt}
\usepackage{graphics}
\usepackage{amsthm}
\usepackage[T1]{fontenc}
\hypersetup{colorlinks=false}

\algnewcommand{\IfThenElse}[3]{
  \State \algorithmicif\ #1\ \algorithmicthen\ #2\ \algorithmicelse\ #3}

\title{Faster Algorithm for Second (s,t)-mincut and Breaking Quadratic barrier for Dual Edge Sensitivity for (s,t)-mincut}

\author{
    Surender Baswana \thanks{Indian Institute of Technology Kanpur, India. Email: \texttt{sbaswana@cse.iitk.ac.in}. Partially supported by Tapas Mishra Memorial Chair at IIT Kanpur, India.}
\and 
    Koustav Bhanja\thanks{Weizmann Institute of Science, Israel. Email: \texttt{koustav.bhanja@weizmann.ac.il}. Partially supported by Merav Parter's European Research Council (ERC) grant under the European Union‚Äôs Horizon 2020 research and
    innovation programme, grant agreement No. 949083.}
\and
    Anupam Roy \thanks{Indian Institute of Technology Kanpur, India. Email: \texttt{anupam@cse.iitk.ac.in}.}
}

\date{}

\begin{document}

\begin{titlepage}
\maketitle
\begin{abstract}
Let $G$ be a directed graph on $n$ vertices and $m$ edges. In this paper, we study $(s,t)$-cuts of second minimum capacity and present the following algorithmic and graph-theoretic results.

    \noindent
    \textbf{1. Second (s,t)-mincut:}  Vazirani and Yannakakis [ICALP 1992] designed the first algorithm for computing an $(s,t)$-cut of second minimum capacity using ${\mathcal O}(n^2)$ maximum $(s,t)$-flow computations. We present the following algorithm that improves the running time significantly.
     For directed integer-weighted graphs, there is an algorithm that can compute an $(s,t)$-cut of second minimum capacity using ${\Tilde{\mathcal O}}(\sqrt{n})$ maximum $(s,t)$-flow computations with high probability.\footnote[1]{$\Tilde{\mathcal{O}}(.)$ hides poly-logarithmic factors.} To achieve this result, a close relationship of independent interest is established between $(s,t)$-cuts of second minimum capacity and global mincuts in directed weighted graphs.

\noindent
        \textbf{2.  Minimum+1 (s,t)-cuts:} 
 Minimum+1 $(s,t)$-cuts have been studied quite well recently [Baswana, Bhanja, and Pandey, ICALP 2022 \& TALG 2023], which is a special case of second $(s,t)$-mincut. We present the following structural result and the first nontrivial algorithm for minimum+1 $(s,t)$-cuts.
 
 \noindent
        (a) Algorithm: For directed multi-graphs, we design an algorithm that, given any maximum $(s,t)$-flow, computes a minimum+1 $(s,t)$-cut, if it exists, in ${\mathcal O}(m)$ time. 

\noindent
(b) Structure: The existing structures for storing and characterizing all minimum+1 $(s,t)$-cuts occupy ${\mathcal O}(mn)$ space [Baswana, Bhanja, and Pandey, TALG 2023]. 
    For undirected multi-graphs, we design a directed acyclic graph (DAG) occupying only ${\mathcal O}(m)$ space that stores and characterizes all minimum+1 $(s,t)$-cuts. 
   This matches the space bound of the widely-known DAG structure for all $(s,t)$-mincuts [Picard and Queyranne, Math. Prog. Studies 1980].   

\noindent
        \textbf{3. Dual Edge Sensitivity Oracle:} 
       The study of minimum+1 (s,t)-cuts often turns out to be useful in designing dual edge sensitivity oracles -- a compact data structure for efficiently reporting an $(s,t)$-mincut after insertion/failure of any given pair of query edges.
       It has been shown recently [Bhanja, ICALP 2025] that any dual edge sensitivity oracle for $(s,t)$-mincut in undirected multi-graphs must occupy ${\Omega}(n^2)$ space in the worst-case irrespective of the query time.
       Interestingly, for simple graphs, we break this quadratic barrier while achieving a non-trivial query time as follows.
       There is an ${\mathcal O}( n\sqrt{n})$ space data structure that can report an $(s,t)$-mincut in ${\mathcal O}(\min\{m,n\sqrt{n}\})$ time after the insertion/failure of any given pair of query edges. 

\noindent
To arrive at our results, as one of our key techniques, we establish interesting relationships between $(s,t)$-cuts of capacity $($minimum$+\Delta)$, $\Delta\ge 0$, and maximum $(s,t)$-flow. We believe that these techniques and the graph-theoretic result in 2.(b) are of independent interest. 
\end{abstract}
\end{titlepage}

\tableofcontents{}
\pagebreak
\pagenumbering{arabic}

\section{Introduction}
 The concept of cut is fundamental in graph theory and has numerous real-world applications \cite{DBLP:books/daglib/0069809}. Let $G=(V,E)$ be a directed weighted graph on $n=|V|$ vertices and $m=|E|$ edges with a designated source vertex $s$ and a designated sink vertex $t$. Every edge $e$ of $G$ is assigned with a non-negative real value as the capacity of the edge, denoted by $w(e)$.
 There are mainly two types of well-studied cuts in a graph -- global cuts and $(s,t)$-cuts. A \textit{global cut}, or simply a \textit{cut}, of $G$ is defined as a nonempty proper subset of $V$. A cut $C$ is said to be an $(s,t)$\textit{-cut} if $s\in C$ and $t\in \overline{C}=V\setminus C$. Every cut of $G$ is associated with a capacity defined as follows. The \textit{capacity} of a cut $C$, denoted by $c(C)$, is the sum of the capacities of every edge $(x,y)$ satisfying $x\in C$ and $y\in \overline{C}$. A global cut (likewise $(s,t)$-cut) of the least capacity is called a \textit{global mincut} (likewise \textit{$(s,t)$-mincut}). Henceforth $\lambda$ denotes the capacity of $(s,t)$-mincut. 
 
 The study of cuts from both structural and algorithmic perspectives has been an important field of research for the past six decades \cite{gomory1961multi, ford_fulkerson_1956, dinitz1976structure, DBLP:journals/jacm/KargerS96, van2023deterministic}. In this paper, we provide the following two main results for $(s,t)$-cuts. ~(1) We present efficient algorithms for computing an $(s,t)$-cut of second minimum capacity in directed weighted graphs. 
 After more than 30 years, our algorithms provide the first improvement by a polynomial factor over the existing result of Vazirani and Yannakakis \cite{vazirani1992suboptimal}. To arrive at this result, we establish that computing an $(s,t)$-cut of second minimum capacity has the same time complexity as computing a global mincut. 
 (2) 
 We present a dual edge sensitivity oracle for $(s,t)$-mincut -- a compact data structure that efficiently reports an $(s,t)$-mincut after the insertion or failure of any pair of edges. This is the first dual edge sensitivity oracle for $(s,t)$-mincut 
 that occupies subquadratic space while achieving nontrivial query time in simple graphs\footnote{A simple graph refers to an undirected unweighted graph having no parallel edges.}. This breaks the existing quadratic lower bounds \cite{baswana2023minimum+, DBLP:conf/icalp/Bhanja25} on the space for any dual edge sensitivity oracle for $(s,t)$-mincut in undirected multi-graphs\footnote{A multi-graph refers to an unweighted graph having parallel edges.}. 
We arrive at this result by designing a compact structure for storing and characterizing all \textit{minimum+1 $(s,t)$-cuts} (a special case of second minimum $(s,t)$-cut). This structure improves the space occupied by the existing best-known structure \cite{baswana2023minimum+} by a linear factor. 

We now present the problems addressed in this paper, their state-of-the-art, and our results.

\subsection{Faster Algorithm for Second Minimum (s,t)-cuts}
An algorithmic graph problem aims at computing a structure that optimizes a given function. Examples of such classical problems are Minimum Spanning Tree, Shortest Path, Minimum Cut. Having designed (near) optimal algorithm for such problems, the next immediate objective is the following.
Design an efficient algorithm that, given a structure $S$ achieving optimal function value, computes a structure $S'$ that differs from $S$ and achieves the optimal or \textit{next} optimal function value. There has been an extensive study on computing the second minimum spanning tree \cite{kMST/camerini1974k,kMSTibaraki,kbst/mayr1992spanning}, the second shortest path \cite{kshortestpath/bock1957algorithm, kshortestpath/yen1971finding, kShortestpath/ED1998,kshortestpath/liamrodity}, and second $(s,t)$-mincut \cite{vazirani1992suboptimal}. The algorithms for these problems are so fundamental that they appear in textbooks on algorithms \cite{cormen2022introduction, DBLP:books/daglib/0069809}. In addition, they act as the foundation for generalizing the problem to compute $k^{th}$ optimal structure, such as computing $k^{th}$ minimum spanning tree, $k^{th}$ shortest path, $k^{th}$ $(s,t)$-mincut.

    The problem of designing efficient algorithms for computing $(s,t)$-mincut (or equivalently, maximum $(s,t)$-flow) has been studied for more than six decades. This problem is now almost settled with the recent breakthrough result on almost linear time algorithm for maximum $(s,t)$-flow by Van et al. \cite{van2023deterministic}. Interestingly, given an algorithm for maximum $(s,t)$-flow, we can compute all $(s,t)$-mincuts implicitly with an additional ${\mathcal O}(m)$ time, as shown by Picard and Queyranne \cite{DBLP:journals/mp/PicardQ80}. It is, therefore, natural to redefine the second $(s,t)$-mincut problem as follows. Design an algorithm that computes an $(s,t)$-cut of second minimum capacity. For brevity, henceforth we call $(s,t)$-cut of second minimum capacity by second $(s,t)$-mincut. 

    Vazirani and Yannakakis \cite{vazirani1992suboptimal} addressed the problem of computing an $(s,t)$-cut having $k^{th}$ minimum capacity in 1992. For computing any $k^{th}$ minimum capacity $(s,t)$-cut, they gave an algorithm that uses ${\mathcal O}(n^{2(k-1)})$ maximum $(s,t)$-flow computations. Conversely, they also argued, using NP-hardness of computing a maximum cut, that exponential dependence on $k$ is unavoidable assuming $P\ne NP$. 
   To arrive at their result, the key problem they addressed is the design of an algorithm that computes a second $(s,t)$-mincut using ${\mathcal O}(n^2)$ maximum $(s,t)$-flow computations.
   They further improve the running time by designing a faster algorithm that computes a second $(s,t)$-mincut using only ${\mathcal O}(n)$ maximum $(s,t)$-flow computations. 
   We show that this faster algorithm is incorrect due to a serious error in its analysis (refer to Section \ref{sec : limitation} for details). As a result, the existing best-known algorithm for computing a second $(s,t)$-mincut uses ${\mathcal O}(n^2)$ maximum $(s,t)$-flow computations \cite{vazirani1992suboptimal}. 
    We design an algorithm for the second $(s,t)$-mincut with a significantly improved running time as shown in the following theorem.
    \begin{theorem} [Second $(s,t)$-mincut algorithm] \label{thm : second minimum (s,t)-cut}
        For any directed graph $G$ on $n$ vertices with integer edge capacities, 
        there is an algorithm that computes a second $(s,t)$-mincut in $G$ using $\Tilde{{\mathcal O}}(\sqrt{n})$ maximum $(s,t)$-flow computations with high probability.
    \end{theorem}
    The two well-known minimum cuts in a graph are $(s,t)$-mincut and global mincut. Recently, there has been a growing interest in understanding the difference in the complexity of computing an $(s,t)$-mincut and computing a global mincut. 
    For undirected weighted graphs, Li and Panigrahi \cite{DBLP:conf/focs/LiP20} established that the running time of computing a global mincut differs by only poly-logarithmic factors from the running time of computing an $(s,t)$-mincut. In particular, the authors showed that there is an algorithm that can compute a global mincut 
    using $\tilde{{\mathcal O}}(1)$ maximum $(s,t)$-flow computations
   with additional $\Tilde{\mathcal{O}}(m)$ time. However, for directed weighted graphs, there is a \textit{large} gap between the running time of computing an $(s,t)$-mincut and computing a global mincut as follows. Cen et al. \cite{DBLP:conf/focs/Cen0NPSQ21} designed an algorithm that can compute a global mincut using $\tilde{\mathcal O}(\sqrt{n})$ maximum $(s,t)$-flow computations with additional $\tilde{\mathcal O}(m\sqrt{n})$ time. Interestingly, to arrive at our result in Theorem \ref{thm : second minimum (s,t)-cut}, 
    we show that the complexity of computing a global mincut is the same as the complexity of computing a second $(s,t)$-mincut \textit{modulo} a single maximum $(s,t)$-flow computation as follows, which might find many other important applications.

    \begin{theorem} [Equivalence between global mincut and second $(s,t)$-mincut] \label{thm : equivalence between second st mincut and global}
        For any directed weighted graph $G$ on $n$ vertices and $m$ edges, the following two assertions hold.
        \begin{enumerate}
            \item The problem of computing a second $(s,t)$-mincut is reducible to the problem of computing a global mincut in ${\mathcal O}(MF(m,n))$ time, where $MF(m,n)$ denotes the time complexity for computing a maximum $(s,t)$-flow in $G$.
            \item The problem of computing a global mincut is reducible to the problem of computing a second $(s,t)$-mincut in $\mathcal{O}(m)$ time. 
        \end{enumerate}
    \end{theorem}
\begin{remark}
    Our algorithm for computing a second $(s,t)$-mincut in Theorem \ref{thm : second minimum (s,t)-cut} is Monte Carlo and works for graphs with integer edge capacities.
    However, the equivalence between second $(s,t)$-mincut and global mincut in Theorem \ref{thm : equivalence between second st mincut and global} is deterministic and holds even for graphs with real edge capacities. 
\end{remark}

\subsection{Minimum+1 (s,t)-cuts: Efficient Algorithm \& Compact Structure}
The cuts of capacity minimum+1, known as minimum+1 cuts, have been studied quite extensively in the past \cite{baswana2023minimum+, DBLP:conf/stoc/DinitzN95, DBLP:conf/icalp/Bhanja25}. For (un)directed multi-graphs, a minimum+1 $(s,t)$-cut is a special case of second $(s,t)$-mincut.  We present the following structural result and the first nontrivial algorithm for minimum+1 $(s,t)$-cuts.
\paragraph*{Algorithm:} For both minimum+1 global cuts \cite{DBLP:conf/stoc/DinitzN95} and minimum+1 $(s,t)$-cuts \cite{baswana2023minimum+}, there exist compact data structures. 
These data structures have important applications in maintaining minimum+2 edge connected components \cite{DBLP:conf/stoc/DinitzN95} and designing dual edge sensitivity oracle for 
$(s,t)$-mincut \cite{baswana2023minimum+}. 
Moreover, there exist efficient algorithms \cite{karger1993, DBLP:conf/focs/Benczur95, smallcutsnagamochi} that can compute a global cut of capacity minimum+1. 
Unfortunately, the existing best-known algorithm for computing an $(s,t)$-cut of capacity minimum+1 is nothing but the algorithm for computing a second $(s,t)$-mincut by \cite{vazirani1992suboptimal}, which uses ${\mathcal O}(n^2)$ maximum $(s,t)$-flow computations.
We present the following result as the first efficient algorithm for computing an $(s,t)$-cut of capacity minimum+1 (proof is in Appendix \ref{sec : minimum+1 algorithm in multigraphs}).  
  \begin{theorem} [Minimum+1 $(s,t)$-cut algorithm] \label{thm : minimum+1 (s,t)-cut}
        For any directed multi-graph $G$ on $n$ vertices and $m$ edges, 
        there is an algorithm that, given any maximum $(s,t)$-flow, computes a minimum+1 $(s,t)$-cut, if exists in $G$, in ${\mathcal O}(m)$ time.
    \end{theorem}
\begin{remark}
    The best-known algorithm for computing an $(s,t)$-mincut uses one maximum $(s,t)$-flow computation. 
    It follows from Theorem \ref{thm : minimum+1 (s,t)-cut} that the running time of computing a minimum+1 $(s,t)$-cut differs from the running time of computing an $(s,t)$-mincut only by additional $\mathcal{O}(m)$ time.
\end{remark}
\paragraph*{Structure:} 
There are several algorithmic applications on cuts, namely, fault-tolerance 
\cite{DBLP:journals/mp/PicardQ80, dinitz1976structure}, dynamic algorithms \cite{henzingericalp2023, DBLP:journals/talg/GoranciHT18}, edge-connectivity augmentation \cite{edgeconnectivityaigmentation/naor1997fast,edgeconnaugmentation/cen2022augmenting} that require an efficient way to distinguish a set of cuts, say minimum cuts or minimum+1 cuts, from the rest of the cuts. 
A trivial way to accomplish this objective is to store each cut of the required set explicitly. However, this is totally impractical since the set of these cuts is usually quite huge. For example, the number of $(s,t)$-mincuts are exponential  \cite{DBLP:journals/mp/PicardQ80}, and the number of global mincuts are $\Omega(n^2)$ \cite{dinitz1976structure}. 
%
%
This has led the researchers to invent the following concept. 
A structure $H$ is said to {\em characterize} a set of cuts ${\mathcal C}$ using a property $P$ if the following condition holds. A cut $C\in {\mathcal C}$ if and only if $C$ satisfies property $P$ in $H$.  
It is desirable that $H$ is as compact as possible. Moreover, verifying if a given cut $C$ satisfies $P$ in $H$ has to be as time-efficient as possible.

The design of compact structures for characterizing minimum cuts started with the seminal work of Dinitz, Karzanov, and Lomonosov \cite{dinitz1976structure} in 1976. In this work, the well-known cactus graph occupying ${\mathcal O}(n)$ space was invented for storing and characterizing all global mincuts. Several compact structures have been designed subsequently for storing and characterizing cuts of capacity both minimum and near minimum \cite{DBLP:journals/mp/PicardQ80, DBLP:conf/focs/Benczur95, DBLP:conf/stoc/DinitzN95, DBLP:journals/siamcomp/DinitzV00,baswana2023minimum+}. Quite expectedly, they are playing crucial roles in establishing many important algorithmic results 
\cite{vazirani1992suboptimal, DBLP:conf/focs/Benczur95, DBLP:conf/stoc/DinitzN95, DBLP:journals/talg/GoranciHT18, DBLP:journals/jacm/KawarabayashiT19, baswana2023minimum+}.
In particular, compact structures for storing and characterizing all minimum+1 cuts of a graph have been well-studied.
For all minimum+1 global cuts in undirected multi-graphs, Dinitz and Nutov \cite{DBLP:conf/stoc/DinitzN95} constructed an ${\mathcal O}(n)$ space 2-level cactus model that stores and characterizes them. 
For all minimum+1 $(s,t)$-cuts in directed multi-graphs, the existing structure that provides a characterization occupies ${\mathcal O}(mn)$ space \cite{baswana2023minimum+}.
Unfortunately, the space occupied by this structure of \cite{baswana2023minimum+} is significantly inferior to the widely-known ${\mathcal O}(m)$ space directed acyclic graph (DAG) of Picard and Queyranne \cite{DBLP:journals/mp/PicardQ80}, which stores and characterizes all $(s,t)$-mincuts using $1$-transversal cuts. An $(s,t)$-cut is said to be \textit{$1$-transversal} if its edge-set\footnote{The \textit{edge-set} of a cut $C$ is the set of edges with exactly one endpoint in $C$.} 
intersects any path at most once.  
Interestingly, we are able to achieve the ${\mathcal O}(m)$ bound on space for storing and characterizing all minimum+1 $(s,t)$-cuts. 


An edge $(u,v)$ is said to \textit{contribute} to a cut $C$ if $u \in C$ and $v \in \overline{C}$.
We first establish that
for any maximum $(s,t)$-flow $f$, there exists a set containing at most $n-2$ edges, called the \textit{anchor} edges, such that for any minimum+1 $(s,t)$-cut $C$, exactly one anchor edge contributes to $C$.
By exploiting this result, we design the following structure for storing and characterizing all minimum+1 $(s,t)$-cuts. 

\begin{theorem} [Structure for minimum+1 (s,t)-cuts] \label{thm: structure for min+1}
    Let $G$ be an undirected multi-graph on $n$ vertices and $m$ edges with a maximum $(s,t)$-flow $f$. There is an ${\mathcal O}(\min\{m,n\sqrt{\lambda}\})$ space structure, consisting of a directed acyclic graph ${\mathcal D}$ and the set of $n-2$ anchor edges, that stores and characterizes all $(s,t)$-mincuts and all minimum+1 $(s,t)$-cuts of $G$ as follows. 
    \begin{enumerate}
        \item  An $(s,t)$-cut $C$ is an $(s,t)$-mincut in $G$ if and only if $C$ is a $1$-transversal 
    cut in ${\mathcal D}$ to which no anchor edge corresponding to $f$ contributes.
        \item  An $(s,t)$-cut $C$ is a minimum+1 $(s,t)$-cut in $G$ if and only if $C$ is a $1$-transversal 
    cut in ${\mathcal D}$ to which exactly one anchor edge corresponding to $f$ contributes.
    \end{enumerate}   
\end{theorem}
For undirected graphs, the best-known structure for storing and characterizing all $(s,t)$-mincuts is the DAG of Picard and Queyranne \cite{DBLP:journals/mp/PicardQ80} that occupies ${\mathcal O}(\min\{m,n\sqrt{\lambda}\})$ space, which is tight as well (refer to \cite{shortlengthversionofmengerstheorem} and Lemma 12 in \cite{henzingericalp2023}). Interestingly, not only our structure in Theorem \ref{thm: structure for min+1} matches the bound on space with the DAG for $(s,t)$-mincuts \cite{DBLP:journals/mp/PicardQ80} but also it stores and characterizes both $(s,t)$-mincuts and minimum+1 $(s,t)$-cuts.


\subsection{Dual Edge Sensitivity Oracle: Breaking the Quadratic Barrier}
The study of minimum+1 $(s,t)$-cuts often turns out to be useful in designing elegant \textit{dual edge sensitivity oracles} \cite{baswana2023minimum+, DBLP:conf/stoc/DinitzN95, DBLP:conf/icalp/Bhanja25}, which is defined as follows. 
\begin{definition}[Dual edge sensitivity oracle] \label{def : dual edge sensitivity}
    A dual edge sensitivity oracle for minimum cuts is a compact data structure that can efficiently report a minimum cut in the graph after the insertion or failure of any given pair of query edges.    
\end{definition}  
Designing sensitivity oracles for various minimum cuts of a graph has been an emerging field of research \cite{DBLP:journals/mp/PicardQ80, DBLP:journals/anor/ChengH91, DBLP:conf/stoc/DinitzN95,DBLP:journals/siamcomp/DinitzV00, 
DBLP:conf/esa/BaswanaGK20, DBLP:conf/soda/BaswanaP22, baswana2023minimum+, baswana2024vital} For $(s,t)$-mincut in multi-graphs, the DAG structure of Picard and Queyranne \cite{DBLP:journals/mp/PicardQ80}, as shown in \cite{baswana2023minimum+}, can be used to design an ${\mathcal O}(n)$ space sensitivity oracle that can report an $(s,t)$-mincut in ${\mathcal O}(n)$ time after the failure/insertion of any single edge. 
It is now interesting to 
design a sensitivity oracle that can handle multiple failures/insertions of edges. To solve this generic problem, as argued in \cite{baswana2023minimum+}, the natural approach is to first design a dual edge sensitivity oracle for $(s,t)$-mincut. This approach has also been taken for various other classical graph problems, including distance and connectivity \cite{DBLP:conf/soda/DuanP09a}, graph traversals \cite{DBLP:conf/podc/Parter15},  reachability \cite{choudhary2016optimal, DBLP:conf/icalp/ChakrabortyC20}.
Moreover, it has been observed that handling two edge failures/insertions is significantly more nontrivial than handling a single one.  It also provides important insights that either expose the hardness or help in generalizing the problem for multiple failures. In order to extend the result of \cite{DBLP:journals/mp/PicardQ80} from single to multiple edge failures/insertions, Baswana, Bhanja, and Pandey \cite{baswana2023minimum+} designed the first dual edge sensitivity oracle for $(s,t)$-mincut occupying ${\mathcal O}(n^2)$ space in (un)directed multi-graphs. It can report a resulting $(s,t)$-mincut in ${\mathcal O}(n)$ time.

Note that quadratic space data structures often pose practical challenges as $n$ can be quite \textit{large} for the real world networks/graphs. It thus raises the need to design sensitivity oracles for various fundamental graph problems that occupy subquadratic space \cite{DBLP:journals/theoretics/BiloCCC0KS24, DBLP:journals/jacm/ThorupZ05, bhanja2024optimal}. Unfortunately, it has been shown \cite{baswana2023minimum+, DBLP:conf/icalp/Bhanja25} that any dual edge sensitivity oracle for $(s,t)$-mincut in undirected multi-graphs must occupy $\Omega(n^2)$ bits of space in the worst case, irrespective of the query time. 
However, for simple graphs, we break this quadratic barrier on the space of any dual edge sensitivity oracle while achieving nontrivial query time as follows.
\begin{theorem} [Dual edge sensitivity oracle for $(s,t)$-mincut] \label{thm : sensitivity oracle for simple graphs}
    Let $G$ be a simple graph on $n$ vertices and $m$ edges. 
    There exists a data structure occupying ${\mathcal O}(\min \{m,n^{1.5}\})$ space that can report an $(s,t)$-mincut $C$ (including the contributing edges of $C$) in ${\mathcal O}(\min \{m,n^{1.5}\})$ time after the failure or insertion of any given pair of query edges in $G$. 
\end{theorem}
For the existing dual edge sensitivity oracles for $(s,t)$-mincut \cite{baswana2023minimum+, DBLP:conf/icalp/Bhanja25}, no nontrivial preprocessing time is known till date. We establish the following almost linear preprocessing time for our dual edge sensitivity oracle in Theorem \ref{thm : sensitivity oracle for simple graphs}. 
%
\begin{theorem} [Preprocessing time]\label{thm : preprocessing of data srtcuture D}
    For simple graphs on $n$ vertices and $m$ edges, there is an algorithm that, given any maximum $(s,t)$-flow, can construct the dual edge sensitivity oracle in Theorem \ref{thm : sensitivity oracle for simple graphs} in ${\mathcal O}(m)$ time.
\end{theorem}
\section{Organization of this Paper}
This paper is organized as follows. Basic preliminaries and notations are in Section \ref{sec : preliminaries}. A detailed overview of our results and the techniques used to arrive at them is provided in Section \ref{sec : overview}. The full version,
containing all the omitted proofs, is provided in the appendix, starting from Appendix \ref{sec: start of appendix}.


\section{Preliminaries}\label{sec : preliminaries}
By integrality of maximum $(s,t)$-flow \cite{ford_fulkerson_1956}, for integer-weighted graphs, we consider any given maximum $(s,t)$-flow to be integral. 
The following notations to be used throughout the paper.
\begin{itemize}
    \item $G \cup A$ (likewise $G\setminus A$): Graph obtained after adding (likewise removing) a set of edges $A$ in $G$.
    \item $(\lambda+\Delta)$ $(s,t)$-cut: An $(s,t)$-cut of capacity $\lambda+\Delta$ where $\Delta\ge 0$.
    \item $(u,v)$-path : A simple directed path from vertex $u$ to vertex $v$.
    \item  $f$ denotes a maximum $(s,t)$-flow in graph $G$.
     \item $c(C,H)$: Capacity of a cut $C$ in a graph $H$. 
    \item A cut $C$ \textit{subdivides} a set of vertices $X$ if $C\cap X\ne \emptyset $ and $\overline{C}\cap X \ne \emptyset$.
    \item A cut $C$ \textit{separates} a pair of vertices $u,v$ if $u \in C$ and $v \in \overline{C}$ or vice-versa.
    \item The \textit{edge-set} of a cut $C$, denoted by $E(C)$, is the set of all edges whose endpoints are separated by $C$.
\end{itemize}
\noindent
Let $f'$ be any $(s,t)$-flow in $G$.
\begin{itemize}
    \item $H^{f'}$ denotes the residual graph for any graph $H$ corresponding to an $(s,t)$-flow $f'$.
    \item $f'(e)$ denotes the value of flow  $f'$ assigned to an edge $e$.
    \item $f'_{out}(C)$ and $f'_{in}(C)$: For any $(s,t)$-cut $C$, $f'_{out}(C)$ (likewise $f'_{in}(C)$) is the sum of flow assigned to all edges $e=(u,v) \in E(C)$ with $u \in C,v\in \overline{C}$ (likewise $v\in C, u \in \overline{C}$).
\end{itemize}

\begin{lemma} [Conservation of flow] \label{lem : flow conservation}
    For any $(s,t)$-cut $C$, $f'_{out}(C)-f'_{in}(C)=value(f')$
\end{lemma}
\begin{lemma}[Sub-modularity of Cuts]\label{lem : submodularity}
    For any $C_1,C_2 \subseteq V$, $c(C_1)+c(C_2) \geq c(C_1 \cup C_2)+c(C_1\cap C_2)$
\end{lemma}
\begin{definition} [Quotient Graph] \label{def : quotient graph}
    A graph $H$ is said to be a quotient graph of $G$ if $H$ is obtained from $G$ by contracting disjoint subsets of vertices into single nodes. 
\end{definition}

\begin{definition} [Quotient Path]
    A path $P_q$ is said to be a quotient path of a path $P$ if $P_q$ is obtained from $P$ by contracting a set of edges in $P$.  
\end{definition}
\noindent
\textbf{Residual graph for undirected multi-graphs:} 
Although the residual graph is defined for directed graphs \cite{ford_fulkerson_1956}, for undirected graphs, we define the residual graph in the following way. Let $e=(x,y)$ be any edge in an undirected multi-graph $H$ with an $(s,t)$-flow $f'$. There exist two edges $(y,x)$ (likewise $(x,y)$) in $H^{f'}$ if $e$ carries flow in the direction $x$ to $y$ (likewise $y$ to $x$); otherwise, there is a pair of edges $(x,y)$ and $(y,x)$ in $H^{f'}$.

\subsection{A DAG structure for storing and characterizing all $(s,t)$-mincuts} \label{sec : construction of Dpq}
In a seminal work, Picard and Queyranne \cite{DBLP:journals/mp/PicardQ80} designed 
a DAG, denoted by $\mathcal{D}_{PQ}(G)$, 
that stores all $(s,t)$-mincuts in $G$ and characterizes them as $1$-transversal cuts.
We now briefly describe the construction of $\mathcal{D}_{PQ}(G)$. 
\paragraph*{Construction of $\mathcal{D}_{PQ}(G):$} Let $G'$ be the graph obtained by contracting each Strongly Connected Component (SCC) of $G^f$ into a single node. Let $\mathbb{T}$ denote the node containing $t$ and $\mathbb{S}$ denote the node containing $s$ in $G'$.
If $G$ is an undirected graph, $\mathcal{D}_{PQ}(G)$ is the graph $G'$.
For directed graphs, $\mathcal{D}_{PQ}(G)$ is obtained by suitably modifying  $G'$ as follows.
For each node $\mu$ reachable from $\mathbb{S}$ in $G'$, $\mu$ is contracted into node $\mathbb{S}$. Likewise, each node $\mu$ that has a path to $\mathbb{T}$, 
is contracted into the node $\mathbb{T}$.
Given any maximum $(s,t)$-flow $f$, $\mathcal{D}_{PQ}(G)$ can be obtained in $\mathcal{O}(m)$ time and has the following property. Without causing ambiguity, we refer $({\mathbb S},{\mathbb T})$-cut in ${\mathcal D}_{PQ}(G)$ by $(s,t)$-cut.
\begin{theorem}[\cite{DBLP:journals/mp/PicardQ80}] \label{thm : dag for st mincut and characterization}
    Let $G$ be any directed weighted graph on $m$ edges with a designated source vertex $s$ and designated sink vertex $t$. There is an ${\mathcal O}(m)$ space DAG ${\mathcal D}_{PQ}(G)$ that stores and characterizes each $(s,t)$-mincut in $G$ as follows. An $(s,t)$-cut $C$ is an $(s,t)$-mincut in $G$ if and only if $C$ is a $1$-transversal cut in ${\mathcal D}_{PQ}(G)$.    
\end{theorem}
Let $V(\mu)$ denote the vertices mapped to a node $\mu$ in $\mathcal{D}_{PQ}(G)$.
A cut $C$ is said to subdivide a node $\mu$ in $\mathcal{D}_{PQ}(G)$ if $C$ subdivides $V(\mu)$. It follows from Theorem \ref{thm : dag for st mincut and characterization} that any $(s,t)$-cut in $G$ that subdivides a node in $\mathcal{D}_{PQ}(G)$ has capacity strictly greater than $\lambda$. Hence, the following lemma is immediate.
\begin{lemma}\label{lem :mapping of nodes in Dpq}
For any pair of vertices $u,v \in V$,
    $u$ and $v$ are mapped to the same node in $\mathcal{D}_{PQ}(G)$ if and only if $u$ and $v$ are not separated by any $(s,t)$-mincut in $G$.
\end{lemma}
It is immediate from Lemma \ref{lem :mapping of nodes in Dpq} and Definition \ref{def : quotient graph} that ${\mathcal D}_{PQ}(G)$ is a quotient graph of $G$.

\section{An Overview of Our Results and Techniques} \label{sec : overview}
We now present an overview of our results and the new techniques applied to arrive at them. 

\subsection{Faster Algorithm for Second (s,t)-mincut}
We design two efficient algorithms for computing a second $(s,t)$-mincut. Our first algorithm computes a second $(s,t)$-mincut for directed weighted graphs using ${\mathcal O}(n)$ maximum $(s,t)$-flow computations, which achieves the aim of Vazirani and Yannakakis \cite{vazirani1992suboptimal}. This algorithm can be seen as an immediate application of the recently invented covering technique of \cite{baswana2023minimum+}.
Our second algorithm takes a totally different approach. This approach is based on a relationship between the second $(s,t)$-mincuts and the global mincuts in directed weighted graphs. 
So, as our main result, we design an algorithm for computing a second $(s,t)$-mincut that uses $\tilde{{\mathcal O}}(\sqrt{n})$ maximum $(s,t)$-flow computations and works for directed graphs with integer edge capacities. 
We now provide an overview of this result.


Observe that a global mincut is not necessarily an $(s,t)$-cut in $G$. On the other hand, any second $(s,t)$-mincut can never be a global mincut in $G$. 
So apparently there does not seem to be any relationship between the global mincuts and the second $(s,t)$-mincuts of $G$. 

Let us first work with a special case when graph $G$ has exactly two $(s,t)$-mincuts -- $\{s\}$ and $V\setminus \{t\}$. Suppose there exists a second $(s,t)$-mincut $C$ in graph $G$ such that $C$ separates all the neighbors of $s$ from all the neighbors of $t$. It is easy to compute a second $(s,t)$-mincut in this graph as follows. Compute a maximum $(s,t)$-flow after contracting all neighbors of $s$ with $s$ and all neighbors of $t$ with $t$. The challenge arises when every second $(s,t)$-mincut has at least one contributing edge that is incident on $s$ and/or $t$. We now show that the residual graph $G^f$ plays a crucial role in overcoming this hurdle. Moreover, $G^f$ turns out to establish the bridge between second $(s,t)$-mincuts and global mincuts.

We begin by stating the following property, which is immediate from the Maxflow-Mincut Theorem.
\begin{property} \label{prop : 1}
    For any graph $\mathbb{G}$ with maximum $(s,t)$-flow $f'$, every $(s,t)$-mincut in $\mathbb{G}$ is an $(s,t)$-cut of capacity zero in $\mathbb{G}^{f'}$.
\end{property}
It follows from Property \ref{prop : 1} 
that there is a bijective mapping between the set of all $(s,t)$-mincuts in $G$ and the set of global mincuts containing $s$ and not $t$ in $G^f$.
However, Property \ref{prop : 1} 
does not reveal any information on how a second $(s,t)$-mincut in $G$ appears in residual graph $G^f$. To explore the structure of second $(s,t)$-mincuts in $G^f$, using only the conservation of flow (Lemma \ref{lem : flow conservation}) and the construction of the residual graph, we provide a generalization of Property \ref{prop : 1} 
as follows (refer to Theorem \ref{thm : min+k in residual graph} in full version). 

\begin{property} \label{prop : p2}
    For any graph $\mathbb{G}$ with maximum $(s,t)$-flow $f'$, every $(s,t)$-cut of capacity $\lambda+\Delta$ in $\mathbb{G}$ appears as an $(s,t)$-cut of capacity $\Delta$ in $\mathbb{G}^{f'}$, where $\Delta\ge 0$.
\end{property}
It follows from Property \ref{prop : p2} 
that every second $(s,t)$-mincut in $G$ is a second $(s,t)$-mincut in $G^f$. 
Let the capacity of second $(s,t)$-mincut in $G$ be $\lambda+\Delta_2$, where $\Delta_2>0$. 
Recall that our aim is to establish a relationship between second $(s,t)$-mincuts and global mincuts using $G^f$. So,
by Property \ref{prop : p2}, 
we need to focus on global cuts of capacity exactly $\Delta_2$ in $G^f$.
However, observe that a global cut of capacity $\Delta_2$ can never be a global mincut in $G^f$ since $\Delta_2> 0$.
Moreover, there may exist \textit{many} global cuts of capacity $\Delta_2$ that cannot be a second $(s,t)$-mincut in $G$.  

It is observed that a directed graph has global mincut capacity strictly greater than zero if it is an SCC. 
It turns out that there is exactly one nontrivial SCC, say $H$, in the residual graph $G^f$ since $G$ has exactly two trivial $(s,t)$-mincuts. 
By exploiting Property \ref{prop : 1} 
and \ref{prop : p2}, we immediately arrive at the following inequality.
\begin{equation}\label{eq : greater}
\textit{The capacity of second $(s,t)$-mincut in $G$ $\ge$ $\lambda+$ the capacity of global mincut in $H$}    
\end{equation}
 Now, we establish the converse of Inequality \ref{eq : greater}. Let us consider any global mincut $A$ in $H$. Observe that $A$ is not a second $(s,t)$-mincut in $G^f$ since $s,t\in \overline{A}$. In fact, the capacity of any global cut $A$ in $H$ might be strictly less than the capacity of $A$ in $G^f$. This is because of the existence of edges that are incident on $s$ and/or $t$, which may contribute to $A$ in $G^f$. 
Interestingly, exploiting the structure of $G^f$, the properties of SCC $H$, and Inequality \ref{eq : greater}, we 
achieve the following bijective mapping (refer to Lemma \ref{lem : second mincut = global mincut for 2 st mincuts} in full version). 
\begin{property} \label{prop : p3}
    Let $C_1$ be a global mincut in $H$ and $C_2$ be a second $(s,t)$-mincut in $G$. Then,
    \begin{enumerate}
        \item capacity of $C_2$ in $G$ $=\lambda+$ the  capacity of $C_1$ in $H$ and 
        \item $C_1\cup \{s\}$ is a second $(s,t)$-mincut in $G$ and $C_2\setminus \{s\}$ is a global mincut in $H$.
    \end{enumerate} 
\end{property}
It turns out that Property \ref{prop : p3} 
does not immediately hold for graphs with one or more than two $(s,t)$-mincuts. For graphs with exactly one $(s,t)$-mincut, we suitably modify the SCC $H$ in $G^f$. Finally, by crucially exploiting the structural properties of the DAG $\mathcal{D}_{PQ}(G)$ (Theorem \ref{thm : dag for st mincut and characterization}), we extend our results to general graphs having any number of $(s,t)$-mincuts. This leads to Theorem \ref{thm : equivalence between second st mincut and global}(1). The proof of Theorem \ref{thm : equivalence between second st mincut and global}(2) is straightforward using standard techniques.  

Cen et al. \cite{DBLP:conf/focs/Cen0NPSQ21} recently designed an efficient algorithm for computing a global mincut in directed graphs with integer edge capacities. Their algorithm uses $\tilde{{\mathcal O}}(\sqrt{n})$ maximum $(s,t)$-flow computations to compute a global mincut. This result of Cen et al. \cite{DBLP:conf/focs/Cen0NPSQ21}, along with Theorem \ref{thm : equivalence between second st mincut and global}, immediately leads to Theorem \ref{thm : second minimum (s,t)-cut}. 

\subsection{Compact Structure for All Minimum+1 $(s,t)$-cuts} \label{sec : overview 2}
We address the following problem for undirected multi-graphs in this section.
\begin{problem}\label{prob: lambda+1 strcuture}
    Design a compact structure for storing and characterizing all $(\lambda+1)$ $(s,t)$-cuts.
\end{problem}
To solve Problem \ref{prob: lambda+1 strcuture}, we establish the following flow-based characterization of $(\lambda+1)$ $(s,t)$-cuts, which is of independent interest.
\paragraph*{Flow based characterization of $(\lambda+1)$ ($s,t$)-cuts:} 
In a seminal work in 1956, Ford and Fulkerson \cite{ford_fulkerson_1956} established a strong duality between $(s,t)$-mincut and maximum $(s,t)$-flow, which is widely known as the \textit{Maxflow-Mincut Theorem}.  
This theorem provides a characterization of all $(s,t)$-mincuts using a maximum $(s,t)$-flow. 
However, it seems that any extension of this result might not exist for characterizing $(s,t)$-cuts of capacity $\lambda+1$. 
This is because no equivalent $(s,t)$-flow is known corresponding to a $(\lambda+1)$ $(s,t)$-cut, as stated in \cite{baswana2023minimum+}. 
Interestingly, we show that, in fact, there exist close relationships between maximum $(s,t)$-flow and $(\lambda+1)$ $(s,t)$-cuts as follows. 
\begin{property} [refer to Theorem \ref{thm : maxflow and minimum+1 cut} in full version] \label{prop : q2} 
    For any undirected multi-graph $G$ with a maximum $(s,t)$-flow $f$, an $(s,t)$-cut $C$ is a $(\lambda+1)$ $(s,t)$-cut in $G$ if and only if there is exactly one edge $e$ in the edge-set of $C$ such that $f(e)=0$.
\end{property}
\paragraph*{Our Solution to Problem \ref{prob: lambda+1 strcuture}:}
The problem of designing a compact structure for storing and characterizing all $(s,t)$-mincuts immediately reduces to Problem \ref{prob: lambda+1 strcuture} by modifying the given graph as follows. Add a dummy source $s'$ (likewise a dummy sink $t'$) with $\lambda-1$ edges from $s'$ to $s$ (likewise $t$ to $t'$). 
Interestingly, for directed multi-graphs, Baswana, Bhanja, and Pandey \cite{baswana2023minimum+} provide a solution to Problem \ref{prob: lambda+1 strcuture} by essentially reducing it to the problem of designing a compact structure for storing and characterizing all $(s,t)$-mincuts. 
However, their structure occupies $\mathcal{O}(mn)$ space.
For undirected multi-graphs, we present a structure for storing and characterizing all
$(\lambda+1)$ $(s,t)$-cuts that occupies $\mathcal{O}(\min\{m,n\sqrt{\lambda}\})$ space as follows.
We construct a graph $G'$ such that every $(\lambda+1)$ $(s,t)$-cut of $G$ appears as an $(s,t)$-mincut in $G'$. This will help us store and characterize all $(\lambda+1)$ $(s,t)$-cuts of $G$ using a structure that stores and characterizes all $(s,t)$-mincuts in $G'$. To achieve this objective, we pursue the following simple idea -- remove one edge from every $(\lambda+1)$ $(s,t)$-cut. A naive implementation of this idea might not work as follows. 
There may be a pair of edges $e,e'$ contributing to the cut $C_1\cap C_2$ defined by the intersection of two $(\lambda+1)$ $(s,t)$-cuts $C_1,C_2$. 
Even if none of the edges $e,e'$ contribute to both $C_1$ and $C_2$, their removal will reduce the $(s,t)$-mincut capacity if $c(C_1\cap C_2)$ is $\lambda$ or $\lambda+1$ (refer to Figure \ref{fig: compact strcuture nontriviality}). 
%
\begin{figure}
 \centering  \includegraphics[width=0.4\textwidth]{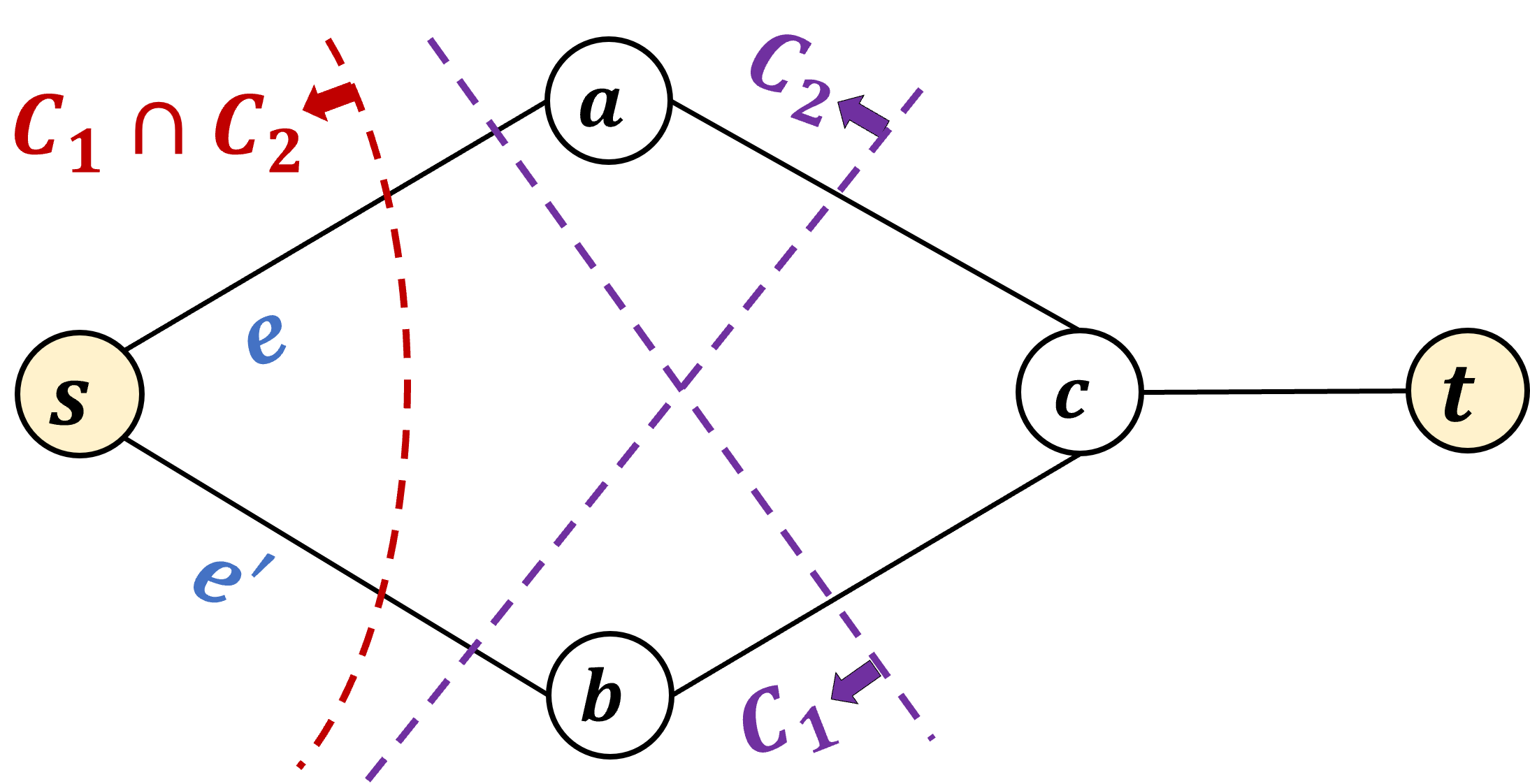} 
  \caption{$C_1\cap C_2$ has capacity less than that of $(s,t)$-mincut after removal of $e,e'$.}
    \label{fig: compact strcuture nontriviality}
\end{figure}
In order to materialize the idea, we exploit \Cref{prop : q2}. This property motivates us to define a set of edges with respect to $(\lambda+1)$ $(s,t)$-cuts in the following way.
\\
\noindent
  \textsc{Anchor Edges:} \textit{An edge is said to be an anchor edge if it does not carry flow in $f$ and contributes to a $(\lambda+1)$ $(s,t)$-cut. }



%
%
By Maxflow-Mincut Theorem, the removal of a set of edges carrying no flow in $f$ does not reduce the capacity of $(s,t)$-mincut. Moreover, by \Cref{prop : q2}, for every $(\lambda+1)$ $(s,t)$-cut $C$ in $G$, exactly one anchor edge contributes to $C$. 
Hence, every $(\lambda+1)$ $(s,t)$-cut, as well as $(s,t)$-mincut in $G$, appears as an $(s,t)$-mincut of capacity $\lambda$ in $G\setminus \mathcal{A}$.
It is also easy to observe that there may exist $(s,t)$-cuts with capacity more than $\lambda+1$ appearing as $(s,t)$-mincuts in $G\setminus {\mathcal A}$. However, using anchor edges ${\mathcal A}$, we can distinguish the $(\lambda+1)$ $(s,t)$-cuts as follows.
\begin{property} \label{prop : q4}
    An $(s,t)$-cut $C$ is a $(\lambda+1)$ $(s,t)$-cut in $G$ if and only if $C$ is an $(s,t)$-mincut in $G\setminus \mathcal{A}$ and exactly one edge from $\mathcal{A}$ contributes to $C$.
\end{property}
By Property \ref{prop : q4}, our compact structure consists of set of edges $\mathcal{A}$ and a structure that stores and characterizes all $(s,t)$-mincuts in $G \setminus \mathcal{A}$. 
It follows that the space occupied by our structure exceeds that of any structure for storing and characterizing all $(s,t)$-mincuts only by the size of ${\mathcal A}$. Therefore, we need to show that the set $\mathcal{A}$ is \textit{small}.
Note that there exist graphs where the number of edges carrying zero flow in a given maximum $(s,t)$-flow can be $\Omega(n^2)$ edges. However, it turns out that the cardinality of set $\mathcal{A}$ is always $\mathcal{O}(n)$ for any given maximum $(s,t)$-flow $f$ in $G$. This is because of the following structural property of anchor edges. 
Any cycle defined by a set of edges carrying zero flow in $f$, cannot contain any anchor edge (refer to Lemma \ref{lem : F has n-2 edges} in full version).
Finally, using the best-known structure $\mathcal{D}_{PQ}$ for storing and characterizing all $(s,t)$-mincuts \cite{DBLP:journals/mp/PicardQ80}, we show that $\mathcal{D}_{PQ}(G \setminus \mathcal{A})$ occupies $\mathcal{O}(\min\{m,n\sqrt{\lambda}\})$ space, which leads to Theorem \ref{thm: structure for min+1}. 

\begin{remark}
    We also show that set $\mathcal{A}$ can be obtained in $\mathcal{O}(m)$ time (refer to Appendix \ref{sec : anchor edge computation}). So, the space bound and preprocessing time of our structure (Theorem \ref{thm: structure for min+1}) match that of the best-known structure for storing and characterizing all $(s,t)$-mincuts \cite{DBLP:journals/mp/PicardQ80}.
\end{remark}
\subsection{Dual Edge Sensitivity Oracle: Breaking the Quadratic Barrier}\label{sec : overview 3}

In this section, for simple graphs, we design a subquadratic space data structure that can efficiently answer the query: report an $(s,t)$-mincut after the failure/insertion of any pair of edges. 
We assume $G$ to be a simple graph in this section. We provide an overview for handling failure of edges in $G$, and handling insertion of edges is along similar lines (refer to Appendix \ref{sec: handling dual edge insertion}). 
Henceforth, let $e_1=(x_1,y_1)$ and $e_2=(x_2,y_2)$ be the two failed edges in $G$. 


There is a folklore result that the residual graph $G^f$ acts as an $\mathcal{O}(m)$ space dual edge sensitivity oracle for $(s,t)$-mincut that achieves $\mathcal{O}(m)$ query time. The query algorithm is derived from the \textit{augmenting path} based algorithm for computing maximum $(s,t)$-flow by Ford and Fulkerson \cite{ford_fulkerson_1956} (briefly explained as a warm-up below). 
However, it is known that the residual graph occupies quadratic space if $m=\Omega(n^2)$. To design a subquadratic space dual edge sensitivity oracle for simple graphs, instead of the residual graph $G^f$, we work with our compact structure, consisting of $\mathcal{D}$ and the set ${\mathcal A}$ of anchor edges, from Theorem \ref{thm: structure for min+1}. Recall that ${\mathcal D}\cup {\mathcal A}$ is just a quotient graph of $G^f$, and hence, it may fail to preserve the complete information of every path in $G^f$.
Even after having this incomplete information, we show that ${\mathcal D}\cup {\mathcal A}$ is still sufficient to answer dual edge failure queries
using the same algorithm used for the folklore result using residual graph $G^f$.  

\paragraph*{Warm-up with residual graph:}  Observe that if none of the failed edges $e_1,e_2$ carry flow in maximum $(s,t)$-flow $f$ then the capacity of $(s,t)$-mincut remains unchanged in $G\setminus \{e_1,e_2\}$. Henceforth, we assume that edge $e_1$ always carries flow in the direction from $x_1$ to $y_1$. It follows from the construction of residual graph that there must exist a $(t,s)$-path $P$ in $G^f$ containing edge $(y_1,x_1)$.
We first reduce the flow in $G$ using $P$ in $G^f$, and then remove the edges $(x_1,y_1)$ and $(y_1,x_1)$ from $G^f$.
Finally, using the concept of \textit{augmenting paths} \cite{ford_fulkerson_1956} in residual graph, we can verify in ${\mathcal O}(m)$ time whether the value of maximum $(s,t)$-flow becomes $\lambda-1$ or remains $\lambda$ in $G\setminus \{e_1\}$.
In the residual graph corresponding to the obtained maximum $(s,t)$-flow in $G\setminus \{e_1\}$, repeat the same procedure for edge $e_2$ to verify whether edge $e_2$ reduces the capacity of $(s,t)$-mincut in $G \setminus \{e_1\}$. 
This helps in reporting an $(s,t)$-mincut in the graph $G\setminus \{e_1,e_1\}$ in ${\cal O}(m)$ time. Complete details are in Appendix \ref{sec: dual edge using residual graph}.


\paragraph*{Our solution using ${\mathcal D}\cup {\mathcal A}$:} 
We now use the structure ${\mathcal D}\cup {\mathcal A}$ from \Cref{thm: structure for min+1} as a subquadratic space dual edge sensitivity oracle.
${\mathcal D}_{PQ}(G)$ (\Cref{thm : dag for st mincut and characterization}) can be used to design a single edge sensitivity oracle that occupies ${\mathcal O}(n)$ space \cite{DBLP:journals/mp/PicardQ80, baswana2023minimum+}.
As observed by Baswana, Bhanja, and Pandey \cite{baswana2023minimum+}, ${\mathcal D}_{PQ}(G)$ can also handle dual edge failures for the special case when both failed edges contribute to only $(s,t)$-mincuts, using its reachability information. However, for handling any dual edge failures,  the main difficulty arises when endpoints of both edges are mapped to the same node in ${\mathcal D}_{PQ}(G)$ \cite{baswana2023minimum+}. 
Our structure ${\mathcal D}\cup {\mathcal A}$ addresses this difficulty seamlessly by first ensuring the following condition. The capacity of $(s,t)$-mincut changes only if the endpoints of at least one failed edge appear in different nodes of ${\mathcal D}$. This is achieved using the following property of $\mathcal{D}$.
\begin{property} [refer to Lemma \ref{lem : mapping in Df} in full version]\label{prop: mapping of nodes in D}
     Any $(s,t)$-cut separating a pair of vertices mapped to the same node in ${\mathcal D}$ must have capacity at least $\lambda+2$. 
\end{property}
Henceforth, we assume without loss of generality that endpoints of edge $e_1$ appear in different nodes in ${\mathcal D}$.
Let us first handle the failure of edge $e_1$. 
It turns out that ${\mathcal D}\cup {\mathcal A}$ can easily report an $(s,t)$-mincut in $G\setminus \{e_1\}$. This exploits the following mapping of paths between $G^f$ and ${\mathcal D}\cup {\mathcal A}$.    
\begin{property} [refer to Lemma \ref{lem : mapping of paths in Df} in full version] \label{prop : r1}
    Let  $u,v$ be any pair of vertices mapped to different nodes $\mu$ and $\nu$ in $\mathcal{D}$. There exists an $(u,v)$-path $P$ in $G^f$ if and only if there exists a $(\mu,\nu)$-path $P_q$ in $\mathcal{D}\cup \mathcal{A}$. Moreover, path $P_q$ is a quotient path of $P$.
\end{property}
We establish Property \ref{prop : r1} 
by using the fact that graph ${\mathcal D}\cup {\mathcal A}$ is a quotient graph of $G^f$ by construction. Let $D_1$ be the graph obtained from ${\mathcal D}\cup {\mathcal A}$ after applying the query algorithm (described above) on ${\mathcal D}\cup {\mathcal A}$ for the failure of edge $e_1$. 
By following the mapping of paths in \Cref{prop : r1}, let $f_1$ be the corresponding maximum $(s,t)$-flow in $G\setminus \{e_1\}$ and 
$G^{f_1}$ denote the corresponding residual graph. On a high level, our technical contribution lies in showing that even after handling failure of $e_1$, $D_1$ still preserves enough information about augmenting paths in $G^{f_1}$ that facilitates the handling of the failure of edge $e_2$ in $G\setminus\{e_1\}$. We now provide the overview.



To handle the failure of edge $e_2$, the obtained graph $D_1$ must satisfy the following property, which actually holds between graphs ${\mathcal D}\cup {\mathcal A}$ and $G^f$ (refer to \Cref{prop: mapping of nodes in D}).
\begin{property} [refer to Lemma \ref{lem : property of D2} in full version] \label{prop : r3}
    $D_1$ is a quotient graph of $G^{f_1}$, and the mapping of paths between them is as follows. For every path $P$ in $D_1$, there is a path $P_1$ in $G^{f_1}$ such that $P$ is a quotient path of $P_1$.
\end{property}
In order to establish \Cref{prop : r3}, the challenging case appears when the $(s,t)$-mincut remains unchanged after the failure of $e_1$. In this case, let us first observe the change in the residual graph after reducing the flow that was passing through edge $e_1$. In the resulting residual graph, the query algorithm finds a path $P$ from $s$ to $t$, and flips the direction of every edge belonging to it. A similar update is executed by the query algorithm using a path $P_1$ in graph ${\mathcal D}\cup {\mathcal A}$, where $P_1$ is a quotient path of $P$. This could lead to the following scenario for the resulting graph $D_1$.
There is a path $P_{uv}$ in $D_1$, but there is no path $P_{uv}^r$ in $G^{f_1}$ such that $P_{uv}$ is a quotient path of $P_{uv}^r$. 
In other words, \Cref{prop : r3} may fail to hold. So, we might report an incorrect $(s,t)$-mincut in $G\setminus \{e_1,e_2\}$. Interestingly, exploiting the SCC structure of $G^{f_1}$, the structure of ${\mathcal D}$, and \Cref{prop: mapping of nodes in D},
we ensure that such a scenario never occurs. This allows us to handle the failure of $e_2$ using $D_1$ exactly in the same way as handling the failure of $e_1$ using ${\mathcal D}\cup \mathcal{A}$ (refer to Lemma \ref{lem : D is same as Dpq} in full version). 
This leads to \Cref{thm : sensitivity oracle for simple graphs}.

\bibliographystyle{alphaurl}
{\small \bibliography{main}}
\appendix
\newpage

\section{Organization of the Full Version} \label{sec: start of appendix}
The full version of this paper is organized as follows. In Appendix \ref{sec : close relationships}, we establish close relationships between maximum $(s,t)$-flow and $(s,t)$-cuts of capacity beyond $(s,t)$-mincut, which are used as tools to arrive at the results in following sections. A limitation of an existing algorithm in \cite{vazirani1992suboptimal} for computing second $(s,t)$-mincut is provided in Section \ref{sec : limitation}. We present two algorithms for computing a second (s,t)-mincut in directed weighted graphs in Appendix \ref{sec : second st mincut in weighted graphs}. For undirected multi-graphs, Appendix \ref{sec : cover and structure} contains the design of our linear space structure for storing and characterizing all $(\lambda+1)$ $(s,t)$-cuts. Finally, in Appendix \ref{section: dual edge oracle}, we design the subquadratic space dual edge sensitivity oracle for $(s,t)$-mincut in simple graphs using the structure constructed in Appendix \ref{sec : cover and structure}.

\section{Minimum+k (s,t)-cuts and Maximum (s,t)-flow} \label{sec : close relationships}
Ford and Fulkerson \cite{ford_fulkerson_1956} established the following strong duality between an $(s,t)$-mincut and a maximum $(s,t)$-flow.
 \begin{theorem} [\text{Maxflow-Mincut Theorem} \cite{ford_fulkerson_1956}] \label{thm : maxflow mincut theorem}
  Let $f$ be any maximum $(s,t)$-flow in $G$. An $(s,t)$-cut $C$ in $G$ is an $(s,t)$-mincut if and only if for every edge $e\in E(C)$, $f(e)=w(e)$ if $e$ is an outgoing edge of $C$ and  $f(e)=0$ if $e$ is an incoming edge of $C$ in $G$.
 \end{theorem}
In this section, as an extension to Theorem \ref{thm : maxflow mincut theorem}, we establish the following two results. 
We first establish a characterization of all $(\lambda+1)$ $(s,t)$-cuts based on a maximum $(s,t)$-flow, which holds only in undirected multi-graphs. Secondly, for directed weighted graphs, we show that for a maximum $(s,t)$-flow $f$, there exists an equivalence between the capacity of an $(s,t)$-cut in $G$ and in $G^{f}$.


\subsection{Undirected Multi-Graphs}

For undirected multi-graphs, we establish the following property for $(\lambda+k)$ $(s,t)$-cuts, where $k\ge 0$ is an integer, based on a maximum $(s,t)$-flow.

\begin{lemma} \label{lem : maxflow and minimum+k cut}
    Let $C$ be a $(\lambda+k)$ $(s,t)$-cut in $G$, where  $k\ge 0$ is an integer. Let $f$ be any maximum $(s,t)$-flow in $G$ and let ${\mathcal E}\subseteq E(C)$ be the set of edges such that $f(e)=0$ for every edge $e\in {\mathcal E}$. Then,
    \begin{enumerate}
        \item $|{\mathcal E}|\le k$,
        \item $|{\mathcal E}|$ is odd if and only if $k$ is odd.  
    \end{enumerate} 
\end{lemma}

\begin{proof}
     Suppose $C$ is a $(\lambda+k)$ $(s,t)$-cut in $G$. It follows from conservation of flow (Lemma \ref{lem : flow conservation}), $f_{out}(C)\ge \lambda$, since $f$ is a maximum $(s,t)$-flow and 
     $f_{in}(C)\ge 0$. 
     Let $f_{out}(C)= \lambda+j$, for any integer $j \ge 0$. Again by Lemma \ref{lem : flow conservation}, $f_{in}(C)=j$. Therefore, we arrive at the following equation.
     \begin{equation} \label{eq : 2}
         f_{out}(C)+f_{in}(C)=\lambda+2j
     \end{equation}
     It follows from Equation \ref{eq : 2} that there are exactly $\lambda+2j$ edges belonging to $E(C)$ that are carrying flow. Let ${\mathcal E}\subseteq E(C)$ be the set of remaining edges such that, for every edge $e\in {\mathcal E}$, $f(e)=0$. 
     In undirected graphs, every edge belonging to the edge-set of a cut is a contributing edge of the cut. Therefore,  using Equation \ref{eq : 2}, we arrive at the following equality.
     \begin{align} \label{equation : undirected graphs}
     \begin{split}
         |{\mathcal E}|&=c(C)-(f_{out}(C)+f_{in}(C))\\
         &=\lambda+k-(\lambda+2j)\\
         &=k-2j
    \end{split}
     \end{align}
     It follows from Equation \ref{equation : undirected graphs} that $|{\mathcal E}|\le k$. In addition, since $2j$ is always an even number, this implies that $|{\mathcal E}|$ is odd if $k$ is odd; otherwise $|\mathcal{E}|$ is even.
\end{proof}
By crucially exploiting Lemma \ref{lem : maxflow and minimum+k cut}, we now establish an interesting \textit{flow-based characterization} of all the $(\lambda+1)$ $(s,t)$-cuts. 
\begin{theorem} [Maxflow (Min+1)-cut Theorem] \label{thm : maxflow and minimum+1 cut}
    Let $G$ be an undirected multi-graph with a designated source $s$ and a designated sink $t$. Let $f$ be any maximum $(s,t)$-flow in $G$. Then, an $(s,t)$-cut $C$ in $G$ is a $(\lambda+1)$ $(s,t)$-cut if and only if there exists exactly one edge $e\in E(C)$ such that
    \begin{enumerate} 
        \item $f(e)=0$ and
        \item for every edge $e'\in E(C)\setminus \{e\}$, $f(e')=1$ and $e'$ carries flow in the direction $C$ to $\overline{C}$.   
    \end{enumerate} 
\end{theorem}
\begin{proof}
    Suppose $C$ is a $(\lambda+1)$ $(s,t)$-cut. Let ${\mathcal E}\subseteq E(C)$ be the set of edges such that $f(e)=0$ for each edge $e\in {\mathcal E}$. By assigning the value of $k=1$ in Lemma \ref{lem : maxflow and minimum+k cut}, we have $|{\mathcal E}|\le 1$ and $|{\mathcal E}|$ is odd. 
    Since $|{\mathcal E}|$ is odd, therefore, $|{\mathcal E}|=1$.
    Since $f$ is a maximum $(s,t)$-flow, by Lemma \ref{lem : flow conservation}, $f_{out}(C)\ge \lambda$. Moreover, we have $|E(C)|=\lambda+1$. Therefore, for every edge $e'\in E(C)\setminus \{e\}$, $f(e')=1$ and $e'$ carries flow in the direction $C$ to $\overline{C}$.        

    We now prove the converse part. Suppose $C$ is an $(s,t)$-cut satisfying properties (1) and (2). It follows that there is no edge $e' \in E(C)$ such that $f(e')=1$ and carries flow in the direction $\overline{C}$ to $C$. So, $f_{in}(C)=0$. By using Lemma \ref{lem : flow conservation}, we have $f_{out}(C)=\lambda$. So, in $E(C)$, there are exactly $\lambda$ edges that are carrying flow and exactly one edge that is carrying no flow. Therefore, $|E(C)|=\lambda+1$.  
\end{proof}
\begin{figure}
 \centering  \includegraphics[width=0.7\textwidth]{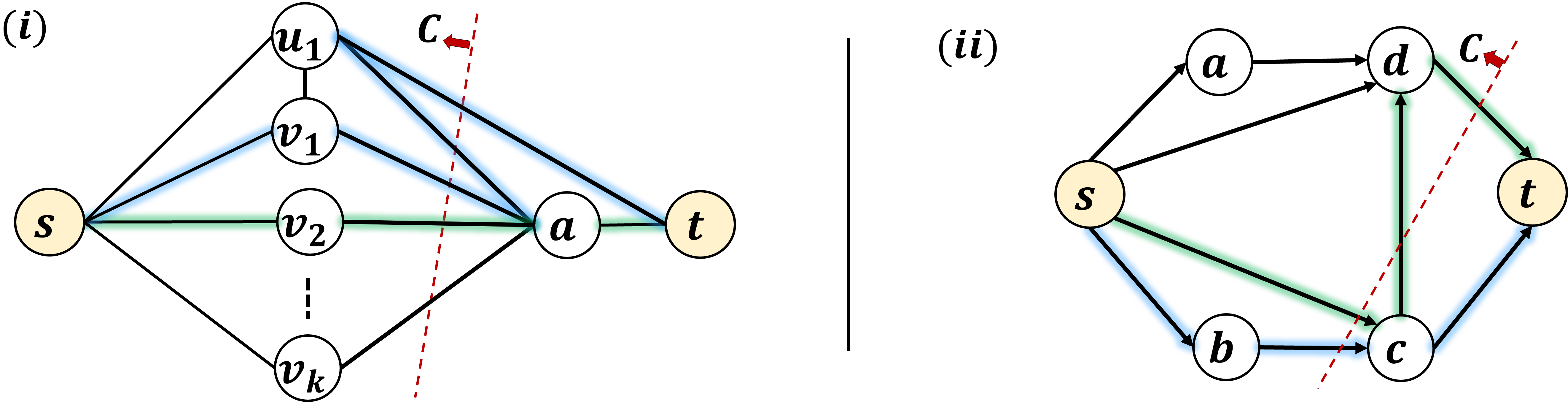} 
  \caption{A colored edge represents that the edge carries flow and $\lambda=2$.
  $(i)$ Example showing that Theorem \ref{thm : maxflow and minimum+1 cut} cannot be generalized for $k \ge 2$
  ($ii$) Cut $C$ is a $(\lambda+1)$ $(s,t)$-cut with $k=1$ where each edge belonging to $E(C)$ carries flow. So, $|\cal E|$ is even although $k$ is odd. }
    \label{fig: counter_example_theorem8_lemma5}
\end{figure}
\begin{remark}
    An immediate generalization of Theorem \ref{thm : maxflow and minimum+1 cut} would be the following. An $(s,t)$-cut $C$ is a $(\lambda+k)$ $(s,t)$-cut if and only if there are exactly $k$ edges in $E(C)$ that carry zero flow and every other edge carries flow in the direction $C$ to $\overline{C}$. 
    However, Figure \ref{fig: counter_example_theorem8_lemma5}$(i)$ shows that 
    for $\lambda=2$, $(s,t)$-cut $C$ has capacity $k+2$ but there are $k-2$ edges in $E(C)$ carrying zero flow.
   Hence, this generalization of Theorem \ref{thm : maxflow and minimum+1 cut} is not possible for any $k\ge 2$.
\end{remark}

\subsection{Directed Weighted Graphs} \label{sec : directed weighted graph tool}
    For directed graphs, Lemma \ref{lem : maxflow and minimum+k cut}$(2)$, as well as Theorem \ref{thm : maxflow and minimum+1 cut}, does not necessarily hold. This is because there exist edges incoming to a cut in directed graphs, and hence, Equation \ref{equation : undirected graphs} fails to satisfy. 
    (refer to Figure \ref{fig: counter_example_theorem8_lemma5}$(ii)$). 
    Observe that the following lemma immediately follows from the Theorem \ref{thm : maxflow mincut theorem} (Maxflow-Mincut Theorem).
    \begin{lemma} \label{lem : alternative maxflow mincut}
        For any maximum $(s,t)$-flow $f$, 
    an $(s,t)$-cut $C$ is an $(s,t)$-mincut if and only if the capacity of $C$ is $0$ in the corresponding residual graph $G^f$.
    \end{lemma}
     Interestingly, for directed weighted graphs, the result in Lemma \ref{lem : alternative maxflow mincut} can be extended to $(\lambda+\Delta)$ $(s,t)$-cuts, where $\Delta \ge 0$ as shown in the following theorem. 
 
\begin{theorem}[\text{Maxflow (Min+$\Delta$)-cut Theorem}] \label{thm : min+k in residual graph} 
    Let $G=(V,E)$ be a directed weighted graph 
    with a designated source vertex $s$ and a designated sink vertex $t$.
    Let $f$ be any maximum $(s,t)$-flow in $G$ and $G^f$ be the corresponding residual graph. Let $C$ be an $(s,t)$-cut in $G$. The capacity of $C$ in $G$ is $(\lambda+\Delta)$ if and only if the capacity of $C$ in $G^f$ is $\Delta$, where $\Delta\ge 0$.    
\end{theorem}
\begin{proof}
    Let $E_{out}(A,H)$ (likewise $E_{in}(A,H)$) denote the set of outgoing edges (likewise the set of incoming edges) of a cut $A$ in a graph $H$.
    For any edge $e$ in $G$, let $r(e)$ be the residual capacity, that is, $r(e)=w(e)-f(e)$.
    For any edge $e$ in graph $G^f$, let $w'(e)$ denote the capacity of edge $e$.
    By construction of $G^f$, for each edge $e=(u,v) \in E_{out}(C,G)$ with $r(e)\neq 0$, there exists a forward edge $(u,v) \in E_{out}(C,G^f)$ with $w'(u,v)=r(e)$. Similarly, for each edge $e=(u,v) \in E_{in}(C,G)$ with $f(e)\ne 0$, there exists a backward edge $(v,u) \in E_{out}(C,G^f)$ with $w'(v,u)=f(e)$.
    This provides us with the following equality. 
   \begin{align}\label{eq : 3}
        \begin{split}
         \sum_{e'\in E_{out}(C,G^f)} w'(e')&=\sum_{e\in E_{out}(C,G)}r(e)+ \sum_{e\in E_{in}(C,G)}f(e)\\ 
         &=\sum_{e\in E_{out}(C,G)}(w(e)-f(e))+ \sum_{e\in E_{in}(C,G)}f(e)\\
         &=\sum_{e\in E_{out}(C,G)}w(e)-\sum_{e\in E_{out}(C,G)}f(e) + \sum_{e\in E_{in}(C,G)}f(e)\\
         &=c(C)-(f_{out}(C)-f_{in}(C))\\
         &=c(C)-\lambda \quad\quad \text{ Using Lemma \ref{lem : flow conservation}}
       \end{split}
    \end{align}
Equation \ref{eq : 3} completes the proof.
\end{proof}

\section{Limitation of the Existing Algorithm for Second (s,t)-mincut} \label{sec : limitation}
 We state here a limitation of the existing  
algorithm given by Vazirani and Yannakakis \cite{vazirani1992suboptimal} to compute a second $(s,t)$-mincut. 
The algorithm for computing an $(s,t)$-cut of $k^{th}$ minimum capacity by \cite{vazirani1992suboptimal} uses ${\mathcal O}(n^{2(k-1)})$ maximum $(s,t)$-flows. So, for $k=2$, the algorithm uses ${\mathcal O}(n^2)$ maximum $(s,t)$-flow computations. In the same article \cite{vazirani1992suboptimal}, they stated the following property. 
 \begin{lemma}[Lemma 3.2(1) in \cite{vazirani1992suboptimal}] 
   Let $G^f$ be the residual graph corresponding to a maximum $(s,t)$-flow $f$ in $G$. Let $H$ be an SCC in $G^f$ and $u,v$ be a pair of vertices in $H$. If the capacity of the least capacity cut separating $\{u\}$ and $\{v\}$ is $k$ in $H$, then the least capacity cut separating $\{s,u\}$ and $\{v,t\}$ has capacity $\lambda+k$ in $G$.
\label{lem: vazirani wrong lemma}
\end{lemma}
\begin{figure}[ht]
 \centering  \includegraphics[width=0.8\textwidth]{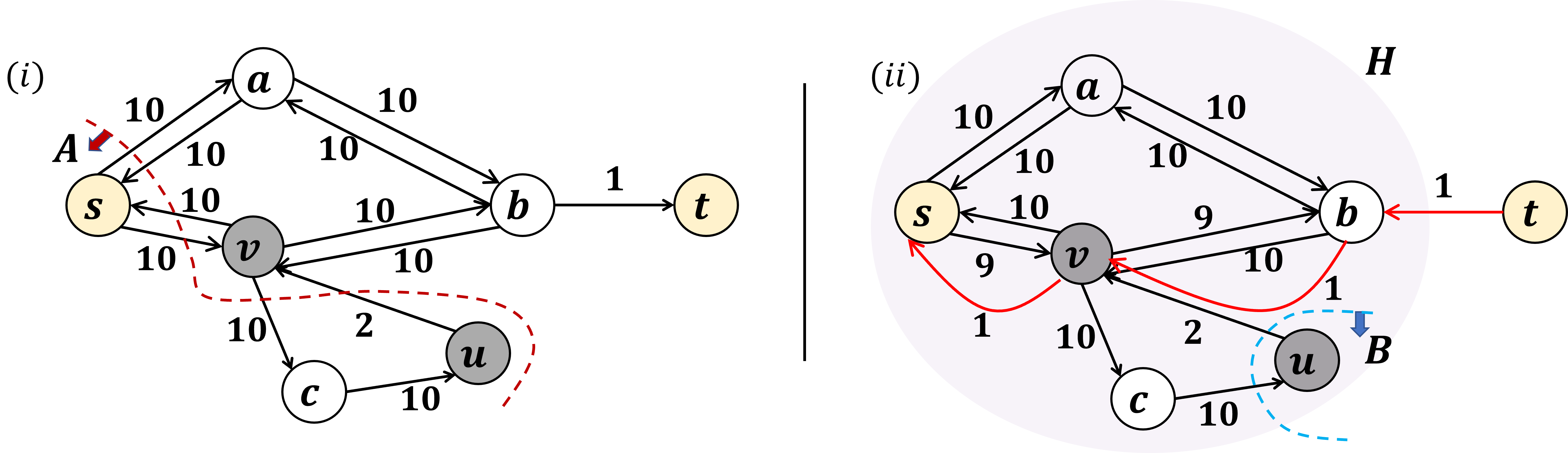} 
  \caption{A counter example for Lemma 3.2(1) in \cite{vazirani1992suboptimal}
  $(i)$ A graph $G$. $V\setminus \{t\}$ is the $(s,t)$-mincut in $G$ with capacity $1$ and $c(A)$=22. $(ii)$ $H$ is an  SCC in $G^f$ and $c(B,H)=2$ 
 }
    \label{fig: counter_examples_VY92}
\end{figure}
 Vazirani and Yannakakis \cite{vazirani1992suboptimal} used Lemma \ref{lem: vazirani wrong lemma} to design an algorithm that computes second $(s,t)$-mincut using ${\mathcal O}(n)$ maximum $(s,t)$-flow computations (the algorithm preceding Theorem 3.3 in \cite{vazirani1992suboptimal}). 
Unfortunately,  Lemma \ref{lem: vazirani wrong lemma} (or Lemma 3.2(1) in \cite{vazirani1992suboptimal}) is not correct as shown in Figure \ref{fig: counter_examples_VY92}$(ii)$. $H$ is the SCC in $G^f$ containing vertices $u$ and $v$. The least capacity cut $B$ separating $\{u\}$ and $\{v\}$ has capacity $2$ in $H$. 
By Lemma \ref{lem: vazirani wrong lemma}, the least capacity cut separating $\{s,u\}$ and $\{v,t\}$ must have capacity $2+1=3$ in $G$. 
However, as it can be verified easily using Figure \ref{fig: counter_examples_VY92}$(i)$, $A$ is a least capacity cut separating $\{s,u\}$ and $\{v,t\}$ in $G$ and its capacity is $22$.
Therefore, the algorithm stated above Theorem 3.3 in \cite{vazirani1992suboptimal}, fails to correctly output a second $(s,t)$-mincut. 

\section{Efficient Algorithms for Computing Second (s,t)-mincut} \label{sec : second st mincut in weighted graphs}
We design two algorithms for computing a second $(s,t)$-mincut. Our first algorithm uses ${\mathcal O}(n)$ maximum $(s,t)$-flow computations. It is based on the covering technique given in \cite{baswana2023minimum+}. As our main result of this section, we design our second algorithm that uses $\Tilde{\mathcal{O}}(\sqrt{n})$ maximum $(s,t)$-flow computations. It is based on a relationship between $(s,t)$-mincuts and global mincuts as follows.

We begin by exploring the relationship between $(s,t)$-mincuts in $G$ and global mincuts in $G^f$.
Let $C$ be any $(s,t)$-mincut in $G$. 
It follows from Lemma \ref{lem : alternative maxflow mincut} 
that the capacity of $C$ is zero in $G^f$. Hence, every $(s,t)$-mincut in $G$ is a global mincut in $G^f$. Conversely, again by Lemma \ref{lem : alternative maxflow mincut}, every global mincut $C'$ in $G^f$ with $s\in C'$ and $t \in \overline{C'}$ is also an $(s,t)$-mincut in $G$. Therefore, there exists a bijective mapping from the set of all $(s,t)$-mincuts in $G$ to a set of global mincuts in $G^f$. Let us now focus on the set of second $(s,t)$-mincuts in $G$. A second $(s,t)$-mincut in $G$, by definition, has capacity strictly greater than $\lambda$. So we need to explore the relationship between 
the set of $(s,t)$-cuts in $G$ having capacity greater than $\lambda$ and the set of global cuts in $G^f$ of capacity greater than 0. 
Any $(s,t)$-cut in $G$ of capacity $\lambda+\Delta$, $\Delta>0$, appears as an $(s,t)$-cut in $G^f$ of capacity $\Delta$ (refer to Theorem \ref{thm : min+k in residual graph}). Therefore, a second $(s,t)$-mincut in $G$ appears as a second $(s,t)$-mincut in $G^f$.
Now, for global cuts of capacity greater than 0 in $G_f$, we crucially use the insight from the following lemma.

\begin{lemma}\label{lem : global mincut in H is non-zero}
    Let $H$ be a directed weighted graph. $H$ is a strongly connected component (SCC) if and only if the capacity of global mincut in $H$ is strictly greater than zero.
\end{lemma}
\begin{proof}
    Suppose $H$ is an SCC. Let $C$ be a global mincut in $H$. Since $H$ is an SCC, for any $u \in C$ and $v \in \overline{C}$, $v$ is reachable from $u$. Therefore, there exists an edge $e$ on the path from $u$ to $v$, that contributes to cut $C$. Hence, $c(C,H) > 0$. 
    
    Suppose the capacity of global mincut is strictly greater than zero in $H$. 
     Assume to the contrary, $H$ is not an SCC. So, there exists a pair of vertices $u,v$ such that $u$ is not reachable from $v$ or vice-versa. Without loss of generality, assume $v$ is not reachable from $u$. Let $U$ be the set of vertices reachable from $u$ in $H$. 
     Since $v \notin U$, $U$ defines a cut in $H$. It follows from the selection of vertices in $U$ that $c(U,H)=0$, a contradiction.
\end{proof}
It follows from Lemma \ref{lem : global mincut in H is non-zero} that each cut in $G_f$ that subdivides an SCC of $G_f$ has capacity greater than 0. There may exist multiple SCCs in the residual graph $G^f$. 
Therefore, observe that there may also exist global cuts in $G^f$ having capacity strictly greater than zero that do not subdivide any SCC. 
However, if graph $G$ has only two $(s,t)$-mincuts $\{s\}$ and $V\setminus \{t\}$, there is only one SCC $H=V\backslash \{s,t\}$ in $G^f$ containing at least two vertices and not containing $s$ or $t$. Observe that, by Lemma \ref{lem : global mincut in H is non-zero}, any global cut $C$ with $s\in C$ and $t\in \overline{C}$ in $G^f$ has capacity strictly greater than zero if and only if $C$ separates at least one pair of vertices in $H$. Therefore, we first work with a graph that has at most two $(s,t)$-mincuts. Finally, exploiting the results for this special case and the structure of ${\mathcal D}_{PQ}(G)$, we extend our results to any general graphs.

\subsection{Graphs with exactly two (s,t)-mincuts} \label{sec : second mincut in graph with two st mincut}

Suppose graph $G$ has exactly two $(s,t)$-mincuts -- $\{s\}$ and $V\setminus \{t\}$. 
We now present two algorithms for computing a second $(s,t)$-mincut in $G$.

\subsubsection*{Algorithm Using ${\mathcal O}(n)$ Maximum $(s,t)$-flow Computations} 



The following algorithm is immediate for computing a second $(s,t)$-mincut if graph $G$ has exactly one $(s,t)$-mincut instead of two. Suppose the $(s,t)$-mincut in $G$ is $V\setminus \{t\}$; otherwise, for the case with only $(s,t)$-mincut $\{s\}$, consider the transpose graph of $G$ after swapping the roles of $s$ and $t$. For every vertex $x\in V\setminus \{s,t\}$, compute an $(s,t)$-mincut using one maximum $(s,t)$-flow in the graph obtained from $G$ by adding an edge $(x,t)$ of infinite capacity. Then, report the $(s,t)$-mincut of the least capacity among all of the computed $(s,t)$-mincuts.
To use this algorithm in graphs with exactly two $(s,t)$-mincut, observe that trivially we need ${\mathcal O}(n^2)$ maximum $(s,t)$-flow computations.

We now extend this algorithm to compute a second $(s,t)$-mincut using ${\mathcal O}(n)$ maximum $(s,t)$-flow computations for graph $G$ with the two $(s,t)$-mincuts. Let $u$ be any vertex other than $s$ and $t$ in $G$. Observe that, for any second $(s,t)$-mincut $C$, either $u\in C$ or $u\in \overline{C}$.
Exploiting this observation, we construct two graphs, $G^I$ and $G^U$, from $G$ as follows. Graph $G^I$ (likewise $G^U$) is obtained by adding an edge $(s,u)$ (likewise $(u,t)$) of infinite capacity. 
Observe that $G^I$ (likewise $G^U$) has only $(s,t)$-mincut $V\setminus \{t\}$ (likewise $\{s\}$).
It follows from the construction of these two graphs that second $(s,t)$-mincut $C$ appears either in $G^I$ or $G^U$. Hence, we can separately compute a second $(s,t)$-mincut using the algorithm mentioned above in the two graphs and report the minimum between them. This leads to the following result.


\begin{lemma} \label{lem : n maxflow}
    Suppose graph $G$ 
    has exactly two $(s,t)$-mincuts -- $\{s\}$ and $V\setminus \{t\}$. There is an algorithm that can compute a second $(s,t)$-mincut in $G$ using ${\mathcal O}(n)$ maximum $(s,t)$-flow computations. 
\end{lemma}
\begin{remark}
The design of the algorithm stated in Lemma \ref{lem : n maxflow} is based on the covering technique of \cite{baswana2023minimum+} (refer to Theorem 3.2 in Section 3 of \cite{baswana2023minimum+}).  
\end{remark}
\subsubsection*{Algorithm Using One Global Mincut Computation}
To further improve the running time, we take an approach that (1) exploits the residual graph $G^f$ instead of the actual graph $G$ and (2) explores the relation between global mincut and second $(s,t)$-mincut in $G^f$.
As discussed before Appendix \ref{sec : second mincut in graph with two st mincut}, there is exactly one SCC $H$ in the residual graph $G^f$ containing 
at least two vertices, that is, $V\setminus \{s,t\}$. 
Note that any global cut $C$ in $H$ may have a capacity strictly less than the capacity of $C$ in $G^f$. This is because of the existence of edges that are incident on $s$ or $t$ (refer to Figure \ref{fig: illustration of the main proof}$(i)$ and $(ii)$). Interestingly, we establish the following bijective mapping between the set of all second $(s,t)$-mincuts in $G^f$ and the set of all global mincuts in $H$. 

\begin{figure}[ht]
 \centering  \includegraphics[width=\textwidth]{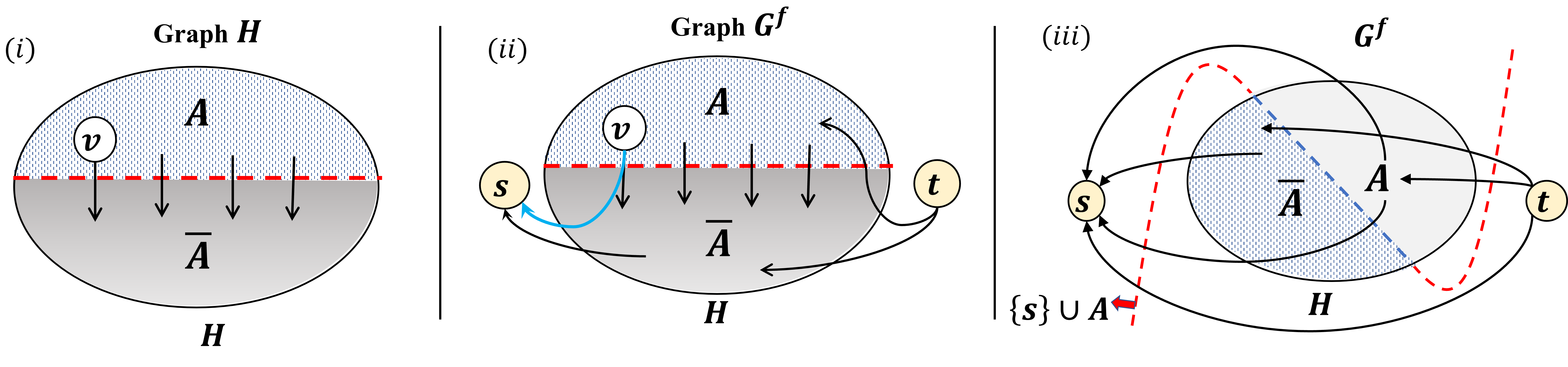} 
  \caption{$(i)$ Global cut $A$ in $H$ $(ii)$ Edge $(v,s)$ contributes to $A$ in $G^f$ but not in $H$ $(iii)$ $A$ is a global mincut in $H$ and $\{s\}\cup A$ is a second $(s,t)$-mincut in $G^f$}
    \label{fig: illustration of the main proof}
\end{figure}


\begin{lemma} \label{lem : second mincut = global mincut for 2 st mincuts}
    Let $C_1$ be a global mincut in $H$ and $C_2$ be a second $(s,t)$-mincut in $G^f$. Then,
    \begin{enumerate}
        \item $c(C_2,G^f)=c(C_1,H)$, 
        \item $C_1\cup \{s\}$ is a second $(s,t)$-mincut in $G^f$ and $C_2\setminus \{s\}$ is a global mincut in $H$.
    \end{enumerate}    
\end{lemma}
\begin{proof}
    Let us consider global mincut $C_1$ in graph $H$, and let $c(C_1,H)=\lambda_1$. 
     By Lemma \ref{lem : global mincut in H is non-zero}, $\lambda_1 >0$.
    It follows from the construction of $G^f$ that
    $C_1 \cup \{s\}$ is an $(s,t)$-cut in $G^f$.
    Observe that the edge-set of $C_1\cup \{s\}$ in $G^f$ and the edge-set of $C_1$ in $H$ differ only on the set of edges, say $E'$, that are incident on $s$ and $t$ in $G^f$. 
    Since $f$ is maximum $(s,t)$-flow, by Lemma \ref{lem : alternative maxflow mincut},
    $s$  has outdegree zero; likewise, $t$ has indegree zero in $G^f$ (refer to Figure \ref{fig: illustration of the main proof}($iii$)). So, every edge in $E'$ is an incoming edge to $(s,t)$-cut $C_1\cup \{s\}$ in $G^f$.
    This implies that the contributing edges of $(s,t)$-cut $C_1 \cup \{s\}$ in $G^f$ are the same as the contributing edges of cut $C_1$ in $H$. Hence, we arrive at the following equation.
    \begin{equation}\label{eq : 4}
        c(C_1 \cup \{s\},G^f)=c(C_1,H)=\lambda_1
    \end{equation} 
     It is given that $C_2$ is a second $(s,t)$-mincut in $G^f$. Let $c(C_2,G^f)=\lambda_2$, where $\lambda_2>0$.
     It follows from  Lemma \ref{lem : alternative maxflow mincut} that $\{s\}$ and $V\setminus \{t\}$ are the only $(s,t)$-mincuts in $G^f$. So, any $(s,t)$-cut in $G^f$ except $\{s\}$ and $V\setminus \{t\}$ has capacity at least $\lambda_2$. 
     Since $C_1$ is a cut in $H$, $ C_1 \cup \{s\}$ is neither $\{s\}$ nor $ V\setminus \{t\}$. Moreover, since
     $C_1 \cup \{s\}$ is an $(s,t)$-cut in $G^f$, we get $c(C_2,G^f) \leq c(C_1\cup \{s\},G^f)$. 
     So, it follows from Equation \ref{eq : 4} that $c(C_2,G^f) \leq \lambda_1$. 
     Therefore, we arrive at the following inequality.
     \begin{equation}\label{eq : 5}
         \lambda_2\le \lambda_1 
     \end{equation}
    Observe that second $(s,t)$-mincut $C_2$ must separate a pair of vertices in $H$ since $\{s\}$ and $V\setminus \{t\}$ are $(s,t)$-mincuts in $G^f$.
    Thus, the cut $C_2\setminus \{s\}$ appears as a global cut in graph $H$. 
    Thus, we can arrive at the following equation using similar arguments as in the case of Equation \ref{eq : 4}. 
    \begin{equation}\label{eq : 6}
        c(C_2\setminus\{s\},H)=c(C_2,G^f)=\lambda_2
    \end{equation}
    Since $C_1$ is a global mincut in $H$ and 
    $C_2\setminus \{s\}$ is a global cut in $H$, $c(C_1,H)\le c(C_2\setminus\{s\},H)$. Therefore, by Equation \ref{eq : 6}, we get the following inequality.
    \begin{equation}\label{eq : 7}
        \lambda_1 \le \lambda_2
    \end{equation}
    
It follows from Equations \ref{eq : 5} and \ref{eq : 7} that $\lambda_2=\lambda_1$. Evidently, $c(C_2,G^f)=c(C_1,H)$ and it completes the proof of (1). Now, by Equation \ref{eq : 4}, $c(C_1\cup \{s\},G^f)=\lambda_2$ and by Equation \ref{eq : 6}, $c(C_2\setminus \{s\},H)=\lambda_1$. Therefore,  $(s,t)$-cut $C_1\cup \{s\}$ is a second $(s,t)$-mincut in $G^f$ and $C_2\setminus \{s\}$ is a global mincut in $H$. This completes the proof of (2).    
\end{proof}
The following lemma is immediate from Lemma \ref{lem : second mincut = global mincut for 2 st mincuts} and Theorem \ref{thm : min+k in residual graph}.
\begin{lemma}
     Suppose $G$ has exactly two $(s,t)$-mincuts -- $\{s\}$ and $V\setminus \{t\}$. There is an algorithm that, given any maximum $(s,t)$-flow in $G$, can compute a second $(s,t)$-mincut in $G$ using one global mincut computation.
     \label{lem : reporting second st mincut in graph with 2 st mincuts}
\end{lemma}

\subsection{Graphs with exactly one (s,t)-mincut}\label{sec : second mincut in graph with one st mincut}
Suppose $G$ has exactly one $(s,t)$-mincut $V\setminus \{t\}$.
Let the capacity of $(s,t)$-mincut in $G$ be $\lambda$.
We now design an algorithm that computes a second $(s,t)$-mincut in graph $G$.
The results can be extended similarly for the case when the only $(s,t)$-mincut in $G$ is $\{s\}$.
Let $H$ denote the graph obtained from $G^f$ by removing $t$ and the edges incident on $t$.

Observe that the algorithm stated in Lemma \ref{lem : n maxflow} actually reduces the problem of computing a second $(s,t)$-mincut in a graph with two $(s,t)$-mincuts to a graph that has exactly one $(s,t)$-mincut. Hence, this algorithm works even for graph $G$ that has exactly one $(s,t)$-mincut. However, the approach taken for the algorithm stated in Lemma \ref{lem : reporting second st mincut in graph with 2 st mincuts} for computing a second $(s,t)$-mincut in graphs with exactly two $(s,t)$-mincuts 
cannot be applied to graph $G$, which has exactly one $(s,t)$-mincut. This is because of the following reasons. Lemma \ref{lem : second mincut = global mincut for 2 st mincuts} crucially exploits the fact that the outdegree of $s$, as well as the indegree of $t$, is zero in the corresponding residual graph. However, for graph $G$, it follows from Theorem \ref{thm : min+k in residual graph} that $s$ has outdegree at least $1$ in $G^f$. In addition, it is not always necessary that $H$ is an SCC. It shows that there may exist global mincuts $C_1$ in $H$ satisfying $s\in \overline{C_1}$ and $c(C_1\cup\{s\}, G^f)\ne c(C_1,H)$. As a result, Equation \ref{eq : 4} fails to hold for graph $G$. Therefore, unlike Lemma \ref{lem : second mincut = global mincut for 2 st mincuts}, it does not seem possible to establish any bijective mapping between the set of all second $(s,t)$-mincuts in $G^f$ and the set of all global mincuts in $H$. To overcome these challenges, we construct a new graph $H_s$ just by providing a simple modification to graph $H$ as follows.
\paragraph*{Construction of $H_s$:}Graph $H_s$ is obtained by adding an edge of infinite capacity from each vertex $v$ to source $s$ in graph $H$. \\

The following fact follows immediately from the construction of $H_s$.
 
\begin{fact} \label{fact : every cut with s on sink side has greater capacity}
    For any pair of cuts $C,C'$ in $H_s$ such that $s\in C$ and $s \in \overline{C'}$, $c(C,H_s) <c(C',H_s)$.
\end{fact}
We now use Fact \ref{fact : every cut with s on sink side has greater capacity} to establish the following lemma.
\begin{lemma} \label{lem : equivalence of global mincut in H and Hs}
    For every global mincut $C$ in $H_s$, $s \in C$ and $c(C,H_s)=c(C,H)$.
\end{lemma}
\begin{proof}
Let $C$ be any global mincut in $H_s$. Fact \ref{fact : every cut with s on sink side has greater capacity} implies that $s \in C$.
Observe that all the edges added during the construction of $H_s$ are either incoming to $C$ or do not belong to the edge-set of $C$ in $H_s$, since $s \in C$.
So, any global cut $C$ in $H_s$ with $s \in C$ has the same set of contributing edges as cut $C$ in $H$. Therefore, $c(C,H_s)=c(C,H)$.
\end{proof}
Since $H$ has only one $(s,t)$-mincut $V\setminus \{t\}$, by construction, $H_s$ is an SCC. So, it follows from Lemma \ref{lem : global mincut in H is non-zero} that the capacity of global mincut in $H_s$ is strictly greater than zero. Therefore, by exploiting Lemma \ref{lem : equivalence of global mincut in H and Hs} and the fact that $t$ has indegree zero in $G^f$, the following lemma can be established along similar lines to the proof of Lemma \ref{lem : second mincut = global mincut for 2 st mincuts}. 

\begin{lemma} \label{lem : second mincut = global mincut for 1 st mincuts}
    Let $C_1$ be a global mincut in $H_s$ and $C_2$ be a second $(s,t)$-mincut in $G^f$. Then,
    \begin{enumerate}
        \item $c(C_2,G^f)=c(C_1,H_s)$, 
        \item $C_1$ is a second $(s,t)$-mincut in $G^f$ and $C_2$ is a global mincut in $H_s$.
    \end{enumerate}    
\end{lemma}
Lemma \ref{lem : second mincut = global mincut for 1 st mincuts} provides a bijective mapping between global mincuts of $H_s$ and second $(s,t)$-mincuts of $G^f$.
Observe that $H_s$ can be obtained from $H$ in $\mathcal{O}(n)$ time. So, by exploiting Lemma \ref{lem : second mincut = global mincut for 1 st mincuts} and Theorem \ref{thm : min+k in residual graph}, we arrive at the following result.
\begin{lemma}
     Suppose graph $G$ has exactly one $(s,t)$-mincut -- either $\{s\}$ or $V\setminus \{t\}$. There is an algorithm that, given any maximum $(s,t)$-flow in $G$, can compute a second $(s,t)$-mincut in $G$ using one global mincut computation.
     \label{lem : reporting second st mincut in graph with 1 st mincuts}
\end{lemma}

\subsection{Extension to General Graphs} \label{sec : extension to egenral graphs}

The algorithms for computing a second $(s,t)$-mincut stated in 
Lemma \ref{lem : reporting second st mincut in graph with 2 st mincuts} and Lemma \ref{lem : reporting second st mincut in graph with 1 st mincuts} work only if the graph has at most two $(s,t)$-mincuts. However, recall that the number of $(s,t)$-mincuts in general graphs can be exponential in $n$. By exploring interesting relations between second $(s,t)$-mincuts and DAG $\mathcal{D}_{PQ}(G)$, we now extend our results for general directed weighted graphs. 
By Theorem \ref{thm : dag for st mincut and characterization}, every $1$-transversal cut of ${\mathcal D}_{PQ}(G)$ is an $(s,t)$-mincut in $G$. It follows that any second $(s,t)$-mincut in $G$ either appears as a non $1$-transversal cut or it subdivides a node in ${\mathcal D}_{PQ}(G)$. In the former case, Vazirani and Yannakakis \cite{vazirani1992suboptimal} showed that the least capacity edge in ${\mathcal D}_{PQ}(G)$ can be used to compute 
a second $(s,t)$-mincut in ${\mathcal O}(m)$ time (refer to Lemma 3.3 and Step 2b in the algorithm preceding Theorem 3.3 in \cite{vazirani1992suboptimal}). However, the problem 
arises in handling the latter case.
Henceforth, we assume that every second $(s,t)$-mincut in $G$ subdivides at least one node of ${\mathcal D}_{PQ}(G)$. The following lemma provides a bound on the number of nodes of ${\mathcal D}_{PQ}(G)$ that any second $(s,t)$-mincut can subdivide. 
\begin{lemma} \label{lem : second st mincut separates one node}
    If $C$ is a second $(s,t)$-mincut in $G$, $C$ subdivides exactly one node in $\mathcal{D}_{PQ}(G)$.
\end{lemma}
\begin{proof}
    Let $\lambda'$ be the capacity of second $(s,t)$-mincut. Assume to the contrary that $C$ is a second $(s,t)$-mincut that subdivides two distinct nodes $\mu_1$ and $\mu_2$ in $\mathcal{D}_{PQ}(G)$. It follows from Lemma \ref{lem :mapping of nodes in Dpq} that there exists an $(s,t)$-mincut $C'$ in $G$ which separates $\mu_1$ and $\mu_2$. Without loss of generality, assume $\mu_1 \in C'$ and $\mu_2 \in \overline{C'}$. By sub-modularity of cuts, $c(C\cap C') +c(C\cup C')\leq \lambda+\lambda'$. Observe that 
    $C\cap C'$ subdivides $\mu_1$ and $C \cup C'$ subdivides $\mu_2$. 
    So, $c(C \cap C'), c(C\cup C') \geq \lambda'$. It follows that $c(C \cap C')+c(C \cup C')\geq 2 \lambda'$. However, $2\lambda' > \lambda+\lambda'$ since $\lambda'>\lambda$, a contradiction. 
\end{proof}
It follows from Lemma \ref{lem : second st mincut separates one node} that the set of all second $(s,t)$-mincuts of $G$ can be partitioned into disjoint subsets as follows. A pair of cuts $C,C'$ belong to different subsets if and only if $C$ and $C'$ subdivide different nodes in ${\mathcal D}_{PQ}(G)$. This partitioning allows us to work separately with each node $\mu$ in $\mathcal{D}_{PQ}(G)$. We now construct a \textit{small} graph $G_\mu$ for node $\mu$ from graph $G^f$ by crucially exploiting the topological ordering of DAG $\mathcal{D}_{PQ}(G)$. 

\paragraph*{Construction of $G_\mu$:} Let $\tau$ be a topological ordering of nodes in $\mathcal{D}_{PQ}(G)$ that begins with ${\mathbb T}$ and ends with ${\mathbb S}$. Graph $G_{\mu}$ is obtained by modifying graph $G^f$ as follows. The set of vertices mapped to the nodes that precede $\mu$ in $\tau$ is contracted into a sink vertex $t'$. Similarly, the set of vertices mapped to the nodes that succeed $\mu$ in $\tau$ is contracted into source vertex $s'$. If $\mu=\mathbb{S}$ (likewise $\mathbb{T}$),
then we map exactly $s$ to $s'$ (likewise $t$ to $t'$).\\ 

It follows from the construction of $G_{\mu}$ that 
$s$ is mapped to $s'$ and $t$ is mapped to $t'$ in $G_\mu$. 
Henceforth, without causing ambiguity, we denote an $(s',t')$-cut in $G_\mu$ by an $(s,t)$-cut. Let $S$ and $T$ be the set of vertices in $V$ mapped to $s'$ and $t'$ respectively.
The following lemma is immediate from the construction of graph $G_\mu$ and Theorem \ref{thm : dag for st mincut and characterization}.
\begin{lemma} \label{lem : at most two st mincuts}
    In graph $G_{\mu}$, the capacity of $(s,t)$-mincut is zero and there are at most two $(s,t)$-mincuts, namely, $\{s'\}$ and $\{s'\} \cup V(\mu)$.
\end{lemma}
The following lemma immediately follows from Lemma \ref{lem : at most two st mincuts}, Lemma \ref{lem : reporting second st mincut in graph with 2 st mincuts}, and Lemma \ref{lem : reporting second st mincut in graph with 1 st mincuts}.
\begin{lemma} \label{lem : second st mincut in G mu}
    There is an algorithm that can compute a second $(s,t)$-mincut in $G_\mu$ using one global mincut computation in $G_\mu$.
\end{lemma}
Now, given a second $(s,t)$-mincut in $G_{\mu}$, we aim to efficiently report a second $(s,t)$-mincut in graph $G$.
To achieve this goal, in the following lemma, we establish an equivalence between second $(s,t)$-mincuts in $G_{\mu}$ and $(s,t)$-cuts of the least capacity in $G^f$ that subdivides $\mu$. 

\begin{lemma} \label{lem : mapping of a second st mincut to a node mu}
    Let $\mu$ be a node in $\mathcal{D}_{PQ}(G)$. Let $C_1$ be an $(s,t)$-cut of the least capacity in $G^f$ that subdivides $\mu$ into $A_1$ and $V(\mu)\setminus A_1$ and let $C_2=\{s'\} \cup A_2$ be a second $(s,t)$-mincut in $G_\mu$, where $A_1, A_2\subset V(\mu)$. Then,
    \begin{enumerate}
         \item $c(C_1,G^f)=c(C_2,G_\mu)$ and
         \item $\{s'\}\cup A_1$ is a second $(s,t)$-mincut in $G_\mu$ and $S\cup A_2$ is an $(s,t)$-cut of the least capacity in $G^f$ that subdivides $\mu$.
    \end{enumerate}
\end{lemma}
\begin{proof}
     Suppose $\mu \neq \mathbb{S}$ or $\mathbb{T}$. 
     Let $c(C_1,G^f)=\lambda_1$ and $c(C_2,G_\mu)=\lambda_2$.
     Observe that $C_1$ can subdivide set $S$, $T$, or both. 
     Let us consider the case when $C_1$ subdivides only $S$. 
    Any prefix of a topological ordering of the nodes in $\mathcal{D}_{PQ}(G)$ defines a $1$-transversal cut in $\mathcal{D}_{PQ}(G)$. By Theorem \ref{thm : dag for st mincut and characterization} and Lemma \ref{lem : alternative maxflow mincut}, $S$ and $S \cup V(\mu)$ are $(s,t)$-mincuts in $G^f$. 
    Since $C_1\cap S$ is an $(s,t)$-cut in $G^f$, $c(C_1\cap S, G^f) \geq c(S,G^f)$. 
    Again $C_1 \cup S$ is an $(s,t)$-cut that subdivides $V(\mu)$ in $G^f$. So, $c(C_1 \cup S,G^f) \geq \lambda_1$. 
    By using sub-modularity of cuts (Lemma \ref{lem : submodularity}) on $(s,t)$-cuts $C_1$ and $S$ in $G^f$, $c(C_1\cup S,G^f)=c(C_1,G^f)$. Since $C_1$ does not divide $T$, $C_1\cup S=S\cup A_1$.
    Now, by construction of $G_\mu$, $S \cup A_1$ appears as the $(s,t)$-cut $\{s'\}\cup A_1$ in $G_\mu$ such that $ c(\{s'\}\cup A_1,G_{\mu})=c(S\cup A_1,G^f)=\lambda_1$.  
    By Lemma \ref{lem : at most two st mincuts}, $\{s'\}\cup A_1$ is not an $(s,t)$-mincut in $G_\mu$ since $\emptyset \ne A_1\subset V(\mu)$. 
    Hence, we arrive at the following inequality.
    \begin{equation}\label{eq: 9}
        \lambda_1\ge \lambda_2
    \end{equation}
    Suppose $C_1$ subdivides only $T$. In this case, using $(s,t)$-mincut $V\setminus T$, we can establish $\lambda_1\ge \lambda_2$ along a similar line to the proof of the case when $C_1$ subdivides only $S$. 
    Suppose $C_1$ subdivides both $S$ and $T$. It follows from the proof of Equation \ref{eq: 9} (for the case when $C_1$ subdivides only $S$) that $C_1\cup S$ subdivides only $T$. Hence, this case reduces to the case when $C_1$ subdivides only $T$.  
  
    Now, consider the $(s,t)$-cut $C_2=\{s'\} \cup A_2$ in $G_\mu$. By construction of $G_\mu$, $C_2$ appears as the $(s,t)$-cut $S\cup A_2$ in $G^f$ such that $c(S\cup A_2,G_{f})=c(C_2,G_\mu)=\lambda_2$.
    Since, $\emptyset\ne A_2\subset V(\mu)$, $S\cup A_2$ is an $(s,t)$-cut that subdivides $\mu$ in $G^f$. Hence, $c(S\cup A_2,G^f)\ge \lambda_1$ and we arrive at the following inequality.
    \begin{equation}\label{eq: 11}
        \lambda_2\ge \lambda_1
    \end{equation}
    It follows from Equations \ref{eq: 9} and \ref{eq: 11} that $\lambda_1=\lambda_2$. 
    This completes the proof. 
\end{proof}

\begin{remark}
    Baswana, Bhanja, and Pandey \cite{baswana2023minimum+} established the result stated in Lemma \ref{lem : mapping of a second st mincut to a node mu} only for the special case when the least capacity $(s,t)$-cut $C$ that subdivides node $\mu$ is a $(\lambda+1)$ $(s,t)$-cut (refer to Lemma 5.10 in \cite{baswana2023minimum+}). Hence, Lemma \ref{lem : mapping of a second st mincut to a node mu} can be seen as an extension of Lemma 5.10 in \cite{baswana2023minimum+}.   
\end{remark}
By crucially exploiting Lemma \ref{lem : mapping of a second st mincut to a node mu} and the algorithm designed for computing a second $(s,t)$-mincut in graph $G_{\mu}$ (Lemma \ref{lem : second st mincut in G mu}), we now state our algorithm for computing a second $(s,t)$-mincut in graph $G$. The pseudocode of the algorithm is given in Algorithm \ref{alg : second mincut using global mincut}.
\paragraph*{Algorithm:} The algorithm begins by computing a topological ordering $\tau$ of the nodes of DAG ${\mathcal D}_{PQ}(G)$ using one maximum $(s,t)$-flow computation in $G$. For each node $\mu\ne \mathbb S, \mathbb T$ in ${\mathcal D}_{PQ}(G)$, compute a global mincut in the SCC $H$ corresponding to node $\mu$; otherwise compute a global mincut in the graph $H_s$ as stated in Lemma \ref{lem : second mincut = global mincut for 1 st mincuts}. Let $\mu$ be a node in ${\mathcal D}_{PQ}(G)$ such that the global mincut $C$ computed for $\mu$ has the least capacity among all global mincuts computed for other nodes in ${\mathcal D}_{PQ}(G)$. 
Let $S$ be the set of vertices mapped to nodes in ${\mathcal D}_{PQ}(G)$ that precedes node $\mu$ in $\tau$. It follows from Lemma \ref{lem : second st mincut in G mu} and the construction of $G_{\mu}$ that $S\cup C$ is a second $(s,t)$-mincut in $G_{\mu}$. Therefore, using Lemma \ref{lem : mapping of a second st mincut to a node mu} and Theorem \ref{thm : min+k in residual graph}, the algorithm reports $S\cup C$ as a second $(s,t)$-mincut in $G$ that has capacity $\lambda+c(C,G_{\mu})$.\\

\begin{algorithm}[ht] 
\caption{Computing Second $(s,t)$-mincut in $G$} \label{alg : second mincut using global mincut}
\begin{algorithmic}[1]
\Procedure{\textsc{Second Mincut}($G,f$)}{}
         \State $d \gets \sum_{e \in E}w(e)$, $\lambda_{min} \gets d$, $C \gets \emptyset, p\gets 0$;
         \State Let $\tau$ be the topological ordering of nodes  in $\mathcal{D}_{PQ}(G)$;
        \For{each node $\mu$ in $\mathcal{D}_{PQ}(G)$}
        \State Let $H$ be the SCC corresponding to node $\mu$;
        \If{$s \in V(\mu)$}
        \State $G'\gets$ Add edge $(v,s)$ to $H$ with $w'(v,s)=d$, for each vertex  $v \in V(\mu)\setminus \{s\}$;
        \ElsIf{$t \in V(\mu)$}
        \State $G' \gets$ Add edge $(t,v)$ to $H$ with $w'(t,v)=d$, for each vertex $v \in V(\mu)\setminus \{t\}$;
        \Else
        \State $G'\gets H$;
        \EndIf
        \State $C'\gets$ Compute a global mincut in $G'$;
        \If{$\lambda_{min} > c(C',G')$} 
        \State Assign $\lambda_{min} \gets c(C',G')$ and $ C \gets C'$;
        \State $p\gets \tau(\mu)$;
        \EndIf
        \EndFor
        \State Let $S$ be the set of vertices mapped to the suffix of node $\mu=\tau(p)$ in topological ordering $\tau$ \label{line : computation of $S$}
        \IfThenElse {$s \in V(\mu)$}
      { \Return $(C,\lambda+\lambda_{min})$;}
      { \Return $(C \cup S,\lambda+\lambda_{min})$;}%
    \EndProcedure
\end{algorithmic}
\end{algorithm}
We now analyze the running time of our algorithm stated in Algorithm \ref{alg : second mincut using global mincut}.
\paragraph*{Running Time:} Given a maximum $(s,t)$-flow in $G$, we can compute $\mathcal{D}_{PQ}(G)$ and its topological ordering $\tau$ in $\mathcal{O}(m)$ time.
Let $GM(n',m')$ denote the time taken to compute a global mincut in a directed weighted graph with $n'$ vertices and $m'$ edges. Observe that $GM(n',m')=\Omega(m')$ and the SCCs corresponding to the nodes of $\mathcal{D}_{PQ}(G)$ are disjoint from each other. Therefore, it is easy to establish that the overall time taken to compute one global mincut in the SCC corresponding to each node in ${\mathcal D}_{PQ}(G)$ is $\mathcal{O}(GM(n,m))$. 
This completes the proof of Theorem \ref{thm : equivalence between second st mincut and global}(1).

Along similar lines to Algorithm \ref{alg : second mincut using global mincut}, using the algorithm stated in Lemma \ref{lem : n maxflow}, it is possible to establish the following result. There is an algorithm that can compute a second $(s,t)$-mincut in general directed weighted graphs using ${\mathcal O}(n)$ maximum $(s,t)$-flow computations. 

We now establish the following result that completes the proof of Theorem \ref{thm : equivalence between second st mincut and global}(2). 
\begin{lemma} \label{lem : global mincut is as hard as second s,t mincut} 
For any directed weighted graph $G$, there is an algorithm that can compute a global mincut in $G$ using one second $(s,t)$-mincut computation. 
\end{lemma}
\begin{proof}
    Given graph $G$, 
    we add two dummy vertices $s_1$ and $t_1$ to obtain graph $G_1$. 
    Observe that the capacity of $(s_1,t_1)$-mincut in $G_1$ is zero. 
    Based on the capacity of global mincut in $G$,
    the proof can be divided into two cases: $(1)$ the capacity of global mincut in $G$ is zero, and $(2)$ it is strictly greater than zero. 
    For Case $1$, we can compute a global mincut in $G$ by taking the following approach. 
    Any global mincut $C$ in $G$ appears as an $(s_1,t_1)$-mincut $s_1 \cup C$ in $G_1$. So, $G_1$ has at least three $(s_1,t_1)$-mincuts. 
    It follows that there are at least four nodes in $\mathcal{D}_{PQ}(G_1)$. 
    Let $\tau$ be a topological ordering of nodes in $\mathcal{D}_{PQ}(G_1)$ and 
    let $C$ be the set of vertices mapped to the suffix of $\tau$ 
    containing two nodes. Since $c(\{s\},G_1)=0$, report $C\setminus \{s\}$ as the global mincut in $G$.
    For Case $2$, 
    $G_1$ has exactly two $(s_1,t_1)$-mincuts: $\{s_1\}$ and $\{s_1\} \cup V$. So, any second $(s_1,t_1)$-mincut in $G_1$ must subdivide $V$. Suppose the algorithm for computing second $(s_1,t_1)$-mincut in $G_1$ returns $(s_1,t_1)$-cut $C$.
    By Lemma \ref{lem : second mincut = global mincut for 2 st mincuts}$(2)$, we can report $C\setminus \{s_1\}$ as a global mincut in $G$.
\end{proof}
The best-known algorithm for computing a global mincut in directed graphs is given by Cen et al. \cite{DBLP:conf/focs/Cen0NPSQ21} as follows.

\begin{theorem} [Theorem I.1 in \cite{DBLP:conf/focs/Cen0NPSQ21}]\label{thm : global mincut in directed integer graphs}
    Let $G$ be a directed graph on $n$ vertices with integer edge capacities. There exists an algorithm that computes a global mincut in $G$ using $\Tilde{\mathcal{O}}(\sqrt{n})$ maximum $(s,t)$-flow computations with high probability.
\end{theorem}
Algorithm \ref{alg : second mincut using global mincut} uses $\mathcal{O}(n)$ invocations of the algorithm in Theorem \ref{thm : global mincut in directed integer graphs}. By using union bound, it is easy to show that Algorithm \ref{alg : second mincut using global mincut} computes a second $(s,t)$-mincut in $G$ with high probability.
Theorem \ref{thm : global mincut in directed integer graphs} and 
Theorem \ref{thm : equivalence between second st mincut and global}(1) lead to Theorem \ref{thm : second minimum (s,t)-cut}.


\section{Compact Structure for All $(\lambda+1)$ (s,t)-cuts} \label{sec : cover and structure}

In this section, we present a compact structure for storing and characterizing all $(\lambda+1)$ $(s,t)$-cuts for undirected multi-graphs. So, let us consider $G$ to be an undirected multi-graph for this section. The construction of our structure involves the following two steps. In the first step, we design an ${\mathcal O}(m)$ space structure consisting of one DAG and a \textit{special} set of edges. In the final step, we improve the space occupied by this structure to ${\mathcal O}(\min\{m,n\sqrt{\lambda}\})$. Unlike the existing approaches \cite{baswana2023minimum+}, to arrive at our structure, we take a \textit{flow-based} approach that crucially exploits the characterization of $(\lambda+1)$ $(s,t)$-cuts using a maximum $(s,t)$-flow (stated in Theorem \ref{thm : maxflow and minimum+1 cut}). 

\subsection{An ${\mathcal O}(m)$ space Structure}\label{sec : structure for min+1 section}
To design a compact structure for storing and characterizing all $(\lambda+1)$ $(s,t)$-cuts of $G$, our aim is to transform all the $(\lambda+1)$ $(s,t)$-cuts into $(s,t)$-mincuts.
In particular, we want to remove a set of edges $E'$ from $G$ such that every $(\lambda+1)$ $(s,t)$-cut of $G$ becomes an $(s,t)$-mincut in the resulting graph. 
Let $f$ be any given maximum $(s,t)$-flow in $G$.
Observe that by Theorem \ref{thm : maxflow mincut theorem} (Maxflow-Mincut Theorem), the removal of a set of edges carrying zero flow in $f$ does not reduce the capacity of $(s,t)$-mincut. 
Let $\textsc{NoFlow}$ denote the set of all edges that carry zero flow in $f$. It is evident that every $(s,t)$-cut has capacity at least $\lambda$ in $G\setminus \textsc{NoFlow}$.
However, is it guaranteed that every $(\lambda+1)$ $(s,t)$-cut is an $(s,t)$-mincut in $G\setminus \textsc{NoFlow}$? 
To address this question, we now introduce the concept of anchor edges. 
\begin{definition}[Anchor edge] \label{def : anchor edge}
    For any given maximum $(s,t)$-flow $f$, an edge $e$ is said to be an anchor edge if $e$ contributes to a $(\lambda+1)$ $(s,t)$-cut and $f(e)=0$. 
\end{definition}
The following property immediately follows from Definition \ref{def : anchor edge} and Theorem \ref{thm : maxflow and minimum+1 cut} for the set of anchor edges. 
\begin{lemma} \label{lem : unique anchor edge}
    Given any maximum $(s,t)$-flow in $G$, for every $(\lambda+1)$ $(s,t)$-cut $C$, there is exactly one anchor edge that contributes to $C$.
\end{lemma}
%
By using Lemma \ref{lem : unique anchor edge}, it is a simple exercise to show that the set of anchor edges, denoted by ${\mathcal A}$, is unique for maximum $(s,t)$-flow $f$. This helps in establishing the following crucial property that holds even for any superset $E'$ of anchor edges carrying zero flow in $G$.

\begin{lemma}  \label{lem : min+1 as 1-transversal in G minus F}
   Let $E'\subseteq E$ be a set of edges in $G$ such that $\mathcal{A} \subseteq E' \subseteq \textsc{NoFlow}$.
   In $G\setminus E'$, the following properties hold. 
   \begin{enumerate}
       \item The capacity of $(s,t)$-mincut is $\lambda$ in $G\setminus E'$.
    \item Every $(\lambda+1)$ $(s,t)$-cut as well as $(s,t)$-mincut in $G$, is an $(s,t)$-mincut in $G\setminus E'$. 
   \end{enumerate}
\end{lemma}
\begin{proof}
    Since $E'\subseteq \textsc{NoFlow}$, removal of all the edges from $E'$ does not decrease the value of maximum $(s,t)$-flow. Therefore, by Maxflow-Mincut Theorem, the capacity of $(s,t)$-mincut in $G\setminus E'$ is $\lambda$. 
    It follows that every $(s,t)$-mincut in $G$ remains an $(s,t)$-mincut in $G\setminus E'$.
    Since $\mathcal{A} \subseteq E'$, Lemma \ref{lem : unique anchor edge} implies that for every $(\lambda+1)$ $(s,t)$-cut $C$ in $G$,
    exactly one anchor edge is removed from the edge-set of $C$. So, the capacity of every $(\lambda+1)$ $(s,t)$-cut in $G$ is reduced by at least one in $G\setminus E'$. Moreover, it follows from Theorem \ref{thm : maxflow and minimum+1 cut} that there is exactly one edge in $E(C)$ carrying no flow. Since $E'\subseteq \textsc{NoFlow}$, the capacity of every $(\lambda+1)$ $(s,t)$-cut in $G$ is reduced by exactly $1$ in $G\setminus E'$. Therefore, every $(\lambda+1)$ $(s,t)$ in $G$ is an $(s,t)$-mincut in $G\setminus E'$. 
\end{proof}
Let $\mathcal{S}(H)$ be any compact structure for storing and characterizing all the $(s,t)$-mincuts in any graph $H$ using a property, say $\mathcal{P}$.
It follows from Lemma \ref{lem : min+1 as 1-transversal in G minus F} that every $(\lambda+1)$ $(s,t)$-cut in $G$ satisfies property $\mathcal{P}$ in $\mathcal{S}(G\setminus E')$.
However, by Lemma \ref{lem : min+1 as 1-transversal in G minus F}$(2)$, every $(s,t)$-mincut in $G$ also satisfies property $\mathcal{P}$. 
Moreover, there may exist many $(s,t)$-cuts other than the $(\lambda+1)$ $(s,t)$-cuts and $(s,t)$-mincuts in $G$ that have also become $(s,t)$-mincuts in $G\setminus E'$. This is because a $(\lambda+k)$ $(s,t)$-cut may contain exactly $k$ anchor edges (refer to Figure \ref{fig: counter_examples_anchor_edges}$(i)$ for $k=2$).
However, we show using Theorem \ref{thm : maxflow and minimum+1 cut} and Lemma \ref{lem : unique anchor edge} that $E'$ is sufficient to characterize $(\lambda+1)$ $(s,t)$-cuts using ${\mathcal S}(G\setminus E')$ as follows. 


\begin{figure}[ht]
 \centering  \includegraphics[width=\textwidth]{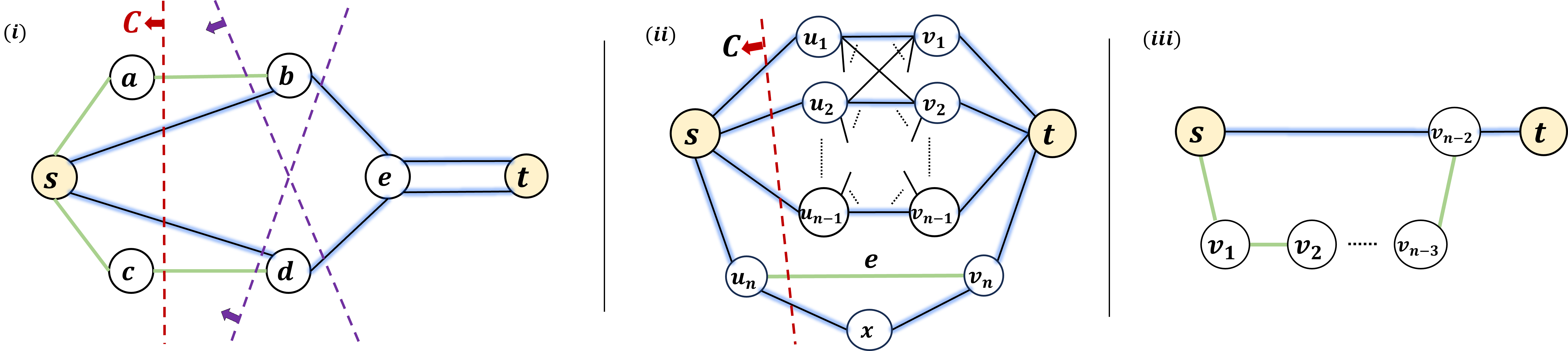} 
  \caption{
A blue (likewise green) edge represents that the edge is carrying nonzero flow (likewise an anchor edge).
  $(i)$ $C$ is a $(\lambda+2)$ $(s,t)$-cut and it contains two anchor edges. 
  ($ii$) Graph satisfying $|\textsc{NoFlow}|=\Omega(n^2)$ with exactly one anchor edge. ($iii$) Graph with exactly $(n-2)$ anchor edges.}
    \label{fig: counter_examples_anchor_edges}
\end{figure}

\begin{lemma}\label{lem :  characterization of min+1 cuts in D(F)}
    Let $E'\subseteq E$ be a set of edges in $G$ such that $\mathcal{A} \subseteq E' \subseteq \textsc{NoFlow}$.
    There is a structure $\mathcal{S}(G\setminus E')$ that stores and characterizes all $(\lambda+1)$ $(s,t)$-cuts and $(s,t)$-mincuts in $G$ using property ${\mathcal P}$ and set of edges $E'$ as follows. 
    \begin{enumerate}
        \item  An $(s,t)$-cut $C$ is an $(s,t)$-mincut in $G$ if and only if $C$ satisfies $\mathcal{P}$ in $\mathcal{S}(G\setminus {E'})$ and no edge in ${E'}$ contributes to $C$. 
        \item  An $(s,t)$-cut $C$ is a $(\lambda+1)$ $(s,t)$-cut in $G$ if and only if $C$ satisfies $\mathcal{P}$ in ${\mathcal S}(G\setminus {E'})$ and exactly one edge of ${E'}$ contributes to $C$.
    \end{enumerate}   
\end{lemma}
\begin{proof}
    Let $C$ be an $(s,t)$-mincut and $C'$ be a $(\lambda+1)$ $(s,t)$-cut in $G$. By Lemma \ref{lem : min+1 as 1-transversal in G minus F}, both $C$ and $C'$ are $(s,t)$-mincuts in graph $G\setminus E'$. Hence, both $C$ and $C'$, satisfy property $\mathcal{P}$ in $\mathcal{S}(G\setminus {E'})$.
    Since $G$ is undirected, by Theorem \ref{thm : maxflow mincut theorem} (Maxflow-Mincut Theorem), 
    there is no edge from \textsc{NoFlow}, and hence from $E'$, that contributes to $(s,t)$-mincut $C$ in $G$. 
    For $(\lambda+1)$ $(s,t)$-cut $C'$, it follows from Theorem \ref{thm : maxflow and minimum+1 cut} that exactly one edge from $\textsc{NoFlow}$ contributes to $C'$, which must be an anchor edge by Lemma \ref{lem : unique anchor edge}. Therefore, exactly one edge from $E'$ contributes to $C'$ since ${\mathcal A}\subseteq E'\subseteq \textsc{NoFlow}$.    

   Let us consider an $(s,t)$ cut $C$ that satisfies property $\mathcal{P}$ in ${\mathcal S}(G\setminus E')$. It follows from from the property of compact structure $\mathcal{S}$ that $C$ has capacity $\lambda$ in $G\setminus E'$. Therefore, if no edge from $E'$ contributes to $C$, then $C$ is an $(s,t)$-mincut in $G$. Similarly, if exactly one edge from $E'$ is contributing to $C$, then $C$ has capacity exactly $\lambda+1$ in $G$.
\end{proof}
The best-known structure for storing and characterizing all $(s,t)$-mincuts is the DAG ${\mathcal D}_{PQ}$ given by Picard and Queyranne \cite{DBLP:journals/mp/PicardQ80} (refer to Section \ref{sec : preliminaries}). For undirected multi-graph $G\setminus E'$, using the result of \cite{shortlengthversionofmengerstheorem}, it is easy to show that there is an integer-weighted graph $G'$ such that the space occupied by the structure ${\mathcal D}_{PQ}(G')$ is ${\mathcal O}(\min\{m,n\sqrt{\lambda}\})$ and it stores and characterizes all $(s,t)$-mincuts of $G\setminus E'$ (the proof is given in Appendix \ref{sec : size of Dpq} for completeness). We stress that if $G$ is a simple graph, then ${\mathcal D}_{PQ}(G\setminus E')$ is the same as ${\mathcal D}_{PQ}(G')$. For simplicity of exposition, without causing any ambiguity, we use ${\mathcal D}_{PQ}(G\setminus E')$ to denote ${\mathcal D}_{PQ}(G')$. So, we use ${\mathcal D}_{PQ}$ as ${\mathcal S}$ to obtain our structure ${\mathcal D}_{PQ}(G\setminus E')$. By construction, structure ${\mathcal D}_{PQ}(G\setminus E')$ also occupies ${\mathcal O}(\min\{m,n\sqrt{\lambda}\})$ space. However, storing set of edges $E'$ can require ${\mathcal O}(m)$ space since $E'\subseteq \textsc{NoFlow}$.


\subsection{Bounding the Cardinality of The Set of Anchor Edges}\label{sec : cover}
In this section, we provide a tight bound on the cardinality of the set of all anchor edges $\mathcal{A}$. 
We present an efficient algorithm that, exploiting the properties of anchor edges, computes a set of edges $\mathcal{F}$ such that the following two properties hold.
\begin{enumerate}
    \item The set of all anchor edges is a subset of $\mathcal{F}$.
    \item The number of edges belonging to $\mathcal{F}$ is at most $n-2$.
\end{enumerate}
By Definition \ref{def : anchor edge}, every anchor edge belongs to set $\textsc{NoFlow}$. 
However, there exist graphs $H$ such that, for any maximum $(s,t)$-flow in $H$, the cardinality of set $\textsc{NoFlow}$ can be $\Omega(n^2)$, yet the number of anchor edges is only $\mathcal{O}(1)$ (refer to Figure \ref{fig: counter_examples_anchor_edges}($ii$)).
So, \textsc{NoFlow} provides a \textit{loose} upper bound on the number of anchor edges. 
Naturally, the question arises whether it is possible to eliminate a large number of edges from \textsc{NoFlow} while still keeping the set of all anchor edges intact. We answer this question in the affirmative by exploiting the following lemma.
\begin{lemma} \label{lem : edge in a cycle has mincut lambda+2}
    Let $\mathbb{H}$ be a cycle formed using a subset of edges in \textsc{NoFlow}.
    Then, no edge in $\mathbb{H}$ can be an anchor edge.
\end{lemma}
\begin{proof}
     Consider any edge $e \in \mathbb{H}$ and $C$ be any $(s,t)$-cut in $G$ such that $e \in E(C)$. 
    Since $C$ is a cut, $C$ must intersect cycle $\mathbb{H}$ at least twice. Hence, $C$ contains at least two edges that carry no flow. Therefore, it follows from Theorem $\ref{thm : maxflow and minimum+1 cut}$ that $C$ cannot be a $(\lambda+1)$ $(s,t)$-cut. 
    Hence, by Definition \ref{def : anchor edge}, $e$ cannot be an anchor edge.
\end{proof}
We now use Lemma \ref{lem : edge in a cycle has mincut lambda+2} to construct a spanning forest $G_\mathcal{F}$ to provide an upper bound on the cardinality of set ${\mathcal A}$.

\paragraph*{Construction of Spanning Forest $G_{\mathcal F}=(V,{\mathcal F})$:} 
The vertex set of $G_{\mathcal F}$ is the same as $G$. Initially, there is no edge in $G_{\mathcal F}$. We construct graph $G_{\mathcal F}$ incrementally by executing the following step for each edge $e$ in \textsc{NoFlow}. 
If a cycle is formed in $G_\mathcal{F}\cup \{e\}$, then by Lemma \ref{lem : edge in a cycle has mincut lambda+2}, $e$ cannot be an anchor edge. Hence, we ignore edge $e$. Otherwise, if there is no cycle in $G_\mathcal{F}\cup \{e\}$, then insert edge $e$ to $G_{\mathcal F}$. 



It follows from the construction of $G_\mathcal{F}$
that $\mathcal{A} \subseteq \mathcal{F}$. Moreover, $|\mathcal{F}|$ is at most $n-1$ since $G_{\mathcal F}$ is a spanning forest. 
We now show a tight bound on $|\mathcal{A}|$, as well as $|\mathcal{F}|$, in the following lemma.

\begin{lemma} \label{lem : F has n-2 edges}
    Set ${\mathcal A}$, as well as set $\mathcal{F}$, contains at most $(n-2)$ edges.
\end{lemma}
\begin{proof}
    Let us assume to the contrary that ${\mathcal F}$ contains exactly $n-1$ edges. 
    It follows that $G_{\mathcal F}$ is a spanning tree. Hence, ${\mathcal F}$ contains at least one edge from the edge-set of every $(s,t)$-cut in $G$. It implies that there exists an $(s,t)$-mincut $C$ in $G$ such that ${\mathcal F}$ contains at least one edge from the edge-set of $C$. 
    By Theorem \ref{thm : maxflow mincut theorem} (Maxflow-Mincut Theorem), for every edge $e \in E(C)$, $f(e)=1$. Thus, ${\mathcal F}$ contains an edge $e$ such that $f(e)=1$, which is a contradiction since $\mathcal{F} \subseteq$ \textsc{NoFlow}. Since $\mathcal{A} \subseteq\mathcal{F}$, it follows that the cardinality of $\mathcal{A}$ is at most $n-2$.
\end{proof}
We show that there also exist graphs where ${\mathcal A}$ contains exactly $n-2$ edges (refer to Figure \ref{fig: counter_examples_anchor_edges}$(iii)$). Therefore, the bound on the set ${\mathcal A}$ given in Lemma \ref{lem : F has n-2 edges} is tight. 

It is easy to show using Union-Find data structure \cite{unionfind1975} that, given any maximum $(s,t)$-flow, the time taken for computing set ${\mathcal F}$ (containing all anchor edges) is ${\mathcal O}(m\alpha(m,n))$, where $\alpha(m,n)$ denotes the inverse Ackermann function. Interestingly, we show that there is an algorithm that, given maximum $(s,t)$-flow $f$, can compute only set $\mathcal{A}$ in $\mathcal{O}(m)$ time (refer to Appendix \ref{sec : anchor edge computation}). This completes the proof of the following Theorem.

\begin{theorem} [Anchor Edges: Cardinality \& Computation] \label{thm : anchor edges}
    Let $G$ be any undirected multi-graph on $n$ vertices and $m$ edges. For any maximum $(s,t)$-flow $f$ in $G$, there is a set containing at most $n-2$ edges of $G$, called the anchor edges, such that for any minimum+1 $(s,t)$-cut $C$ in $G$, exactly one anchor edge contributes to $C$. 
    Moreover, given a maximum $(s,t)$-flow $f$, there is an algorithm that computes all anchor edges for $f$ in ${\mathcal O}(m)$ time. 
\end{theorem}
DAG $\mathcal{D}_{PQ}(G\setminus E')$ occupies $\mathcal{O}(\min\{m,n\sqrt{\lambda}\})$ space (established in Appendix \ref{sec : structure for min+1 section}). Therefore, Lemma \ref{lem :  characterization of min+1 cuts in D(F)} and Theorem \ref{thm : anchor edges} complete the proof of Theorem \ref{thm: structure for min+1}.

\section{Dual Edge Sensitivity Oracle for (s,t)-mincuts} \label{section: dual edge oracle}
In this section, for simple graphs, we design a dual edge sensitivity oracle for $(s,t)$-mincut (refer to Definition \ref{def : dual edge sensitivity}) that occupies subquadratic space while achieving a nontrivial query time. Henceforth, let us consider $G$ to be a simple graph.
As a warm-up, we first explain the folklore result that the residual graph $G^f$ acts as a simple dual edge sensitivity oracle that achieves ${\mathcal O}(m)$ query time. Note that $G^f$ occupies ${\mathcal O}(m)={\mathcal O}(n^2)$ space in the worst case.
In order to break this quadratic barrier for simple graphs, our main result is to show that, the query algorithm used for the residual graph $G^f$ can be applied to DAG ${\mathcal D}_{PQ}(G\setminus {\mathcal A})$ and set of anchor edges ${\mathcal A}$ from Theorem \ref{thm: structure for min+1} even if ${\mathcal D}_{PQ}(G\setminus {\mathcal A})\cup {\mathcal A}$ is just a quotient graph of $G^f$. For simplicity, we denote ${\mathcal D}_{PQ}(G\setminus {\mathcal A})$ by $\mathcal{D}(\mathcal{A})$.

We now state two properties of ${\mathcal D}({\mathcal A})$ that are used crucially in arriving at our results.
Observe that, by Lemma \ref{lem :mapping of nodes in Dpq}, if both endpoints of any edge $e$ are mapped to the same node in ${\mathcal D}_{PQ}(G)$, then the failure of edge $e$ does not reduce the capacity of $(s,t)$-mincut. This property is exploited crucially to show that DAG ${\mathcal D}_{PQ}(G)$ acts as a single edge sensitivity oracle for $(s,t)$-mincut \cite{DBLP:journals/mp/PicardQ80}. 
%
Interestingly, ${\mathcal D}({\mathcal A})$ extends the above property of ${\mathcal D}_{PQ}(G)$ to a pair of edges as follows. 
If all the endpoints of the pair of failed edges are mapped to the same node in ${\mathcal D}({\mathcal A})$, then the $(s,t)$-mincut capacity remains unchanged. This property is a consequence of the following lemma, which follows from Theorem \ref{thm: structure for min+1}.
 \begin{lemma} \label{lem : mapping in Df}
    Let $u$ and $v$ be any pair of vertices in $G$. If $u$ and $v$ are mapped to the same node in ${\mathcal D}({\mathcal A})$, then the $(s,t)$-cut of the least capacity that separates $u$ and $v$ has capacity at least $\lambda+2$.
\end{lemma}
\begin{proof}
    Let $u,v$ be any pair of vertices mapped to the same node in ${\mathcal D}({\mathcal A})$. Suppose there is a $(\lambda+1)$ $(s,t)$-cut $C$ satisfying $u\in C$ and $v\in \overline{C}$. It follows that ${\mathcal D}({\mathcal A})$ fails to store $C$. Therefore, existence of such a cut $C$ is not possible due to Theorem \ref{thm: structure for min+1}.
\end{proof} 

We now state the following relation between the edges of DAG ${\mathcal D}({\mathcal A})$ and maximum $(s,t)$-flow $f$ in $G$, which is immediate from anchor edge definition (Definition \ref{def : anchor edge}), Lemma \ref{lem : every edge has flow in Dpq}, and Theorem \ref{thm: structure for min+1}.
\begin{lemma} \label{lem : every edge has flow in DF}
    For any edge $e=(x,y)$ in ${\mathcal D}({\mathcal A})$, the corresponding undirected edge $e$ in $G$ satisfies $f(e)=1$ and $e$ carries flow in the direction $y$ to $x$. Moreover, for any edge $e\in {\mathcal A}$, $f(e)=0$.    
\end{lemma}
We consider only the failure of edges, the insertion case is given in Appendix \ref{sec: handling dual edge insertion} for better readability. 
Suppose the two failed edges are $e_1=(x_1,y_1)$ and $e_2=(x_2,y_2)$. The failure of any edge carrying no flow does not affect the capacity of $(s,t)$-mincut as $f$ remains unchanged. Henceforth, without loss of generality, assume that $e_1$ always carries flow in the direction $x_1$ to $y_1$. The analysis is along a similar line if $e_1$ carries flow in the direction $y_1$ to $x_1$. 

\subsection{An ${\mathcal O}(m)$ Space Data Structure and ${\mathcal O}(m)$ Query Time} \label{sec: dual edge using residual graph}

The failed edge $e_1$ has been carrying flow of value $1$. So, we first reduce the value of $(s,t)$-flow in $G$ by $1$ as follows. It follows from the construction of $G^f$ that there is at least one $(t,s)$-path $P_1$ in $G^f$ satisfying the following. Path $P_1$ contains the residual edge $(y_1,x_1)$ for edge $e_1$, and for every edge $(v,u)$ in $P_1$, the corresponding edge $(u,v)$ in $G$ carries flow in the direction $u$ to $v$. 
Path $P_1$ can be obtained in ${\mathcal O}(m)$ time by using any traversal algorithm. 
In graph $G\setminus \{e_1\}$, we now obtain a new $(s,t)$-flow $f_1$ of value $\lambda-1$ from $f$. 
For this purpose, we define the following operation, which is also used later for our analysis of breaking the quadratic barrier for dual edge sensitivity oracle for $(s,t)$-mincut in Appendix \ref{sec : handling dual edge in subquadratic space}.\\


\noindent
\textsc{Update\_Path}$(H,P)$: Let $H$ be a directed graph and $P$ be any path in $H$. For every edge $e''=(p,q)$ in $P$, this function \textsc{Update\_Path} removes $e''$ from $H$ and adds an edge $(q,p)$ in $H$.\\

We first update $G^f$ using \textsc{Update\_Path}$(G^f,P_1)$; and then, remove from the resulting graph the pair of edges $(x_1,y_1)$ and $(y_1,x_1)$ corresponding to the failed edge $e_1$ in $G$.
Let $G^{f_1}$ be the obtained graph. Observe that $G^{f_1}$ is the residual graph of $G\setminus \{e_1\}$ for an $(s,t)$-flow $f_1$ of value $\lambda-1$. 
In $G\setminus \{e_1\}$, we want to determine whether $f_1$ is a maximum $(s,t)$-flow. 
To determine this, we use the following lemma, which was established by Ford and Fulkerson \cite{ford_fulkerson_1956}.
\begin{lemma}[\cite{ford_fulkerson_1956}] \label{lem : ford fulkerson augmenting paths}
    For any graph $\mathcal{G}$ with an $(s,t)$-flow $f'$, $f'$ is a maximum $(s,t)$-flow if and only if there is no $(s,t)$-path in the residual graph $\mathcal{G}^{f'}$.
\end{lemma}
It follows from Lemma \ref{lem : ford fulkerson augmenting paths} that in graph $G^{f_1}$, we need to verify whether there is any path $P'$ from $s$ to $t$. Note that at most one $(s,t)$-path may exist; otherwise, it can be shown that the capacity of $(s,t)$-mincut in $G$ is strictly greater than $\lambda$. So, in $G\setminus \{e_1\}$, if $P'$ exists, then, by Lemma \ref{lem : ford fulkerson augmenting paths}, the capacity of $(s,t)$-mincut is $\lambda$; otherwise the capacity of $(s,t)$-mincut is $\lambda-1$. We update path $P'$ in $G^{f_1}$ to obtain a graph $G^{f_2}$ by following the construction of residual graph (refer to Section \ref{sec : preliminaries}); otherwise, $G^{f_1}$ is the same as $G^{f_2}$ if $P'$ does not exist. 
It follows that graph $G^{f_2}$ is the residual graph for a maximum $(s,t)$-flow $f_2$ (of value either $\lambda-1$ or $\lambda$) in graph $G\setminus \{e_1\}$.

We now want to verify whether the failure of $e_2$ reduces the capacity of $(s,t)$-mincut in $G\setminus \{e_1\}$. Interestingly, the following lemma ensures that 
DAG ${\mathcal D}_{PQ}(G\setminus \{e_1\})$ is sufficient for this purpose.
\begin{lemma} [\cite{DBLP:journals/mp/PicardQ80}, \cite{DBLP:conf/soda/BaswanaP22}, and Lemma 4.4 in \cite{baswana2023minimum+}] \label{lem : single edge failure}
     For any undirected multi-graph ${\mathcal G}$, upon failure of any edge $e$ in ${\mathcal G}$, the capacity of $(s,t)$-mincut reduces by $1$ if and only if both endpoints of $e$ are mapped to different nodes of ${\mathcal D}_{PQ}({\mathcal G})$. Moreover, if $(s,t)$-mincut reduces, then, for any topological ordering $\tau$ of the nodes of ${\mathcal D}_{PQ}({\mathcal G})$, the suffix of $\tau$ containing exactly one endpoint of $e$ defines an $(s,t)$-mincut after the failure of edge $e$.
\end{lemma} 
By using the residual graph $G^{f_2}$, we construct ${\mathcal D}_{PQ}(G\setminus \{e_1\})$ in ${\mathcal O}(m)$ time by following the construction given in Section \ref{sec : construction of Dpq}. It follows from Lemma \ref{lem : single edge failure} that we can determine whether $e_2$ contributes to an $(s,t)$-mincut in $G\setminus \{e_1\}$ in ${\mathcal O}(1)$ time. 
Suppose $e_2$ contributes to an $(s,t)$-mincut in $G\setminus \{e_1\}$. If the $(s,t)$-mincut capacity in $G\setminus \{e_1\}$ is $\lambda-1$ (likewise, $\lambda$), then after the failure of two edges $e_1,e_2$, the capacity of $(s,t)$-mincut in $G$ is $\lambda-2$ (likewise, $\lambda-1$). 
In a similar way, we can determine the capacity of $(s,t)$-mincut in $G$ after the failure of two edges $e_1,e_2$ if $e_2$ does not contribute to any $(s,t)$-mincut in $G\setminus \{e_1\}$.  Moreover, again by Lemma \ref{lem : single edge failure}, reporting an $(s,t)$-mincut $C$ in $G\setminus \{e_1\}$ after the failure of edge $e_2$ requires ${\mathcal O}(m)$ time. It is easy to observe that $C$ is also an $(s,t)$-mincut in $G$ after the failure of two edges $e_1,e_2$. 
It leads to the following lemma. 
\begin{lemma} \label{lem : dual edge sensitivity using residual graph}
    There is an ${\mathcal O}(m)$ space data structure that, after the failure of any given pair of query edges, can report an $(s,t)$-mincut and its capacity in ${\mathcal O}(m)$ time.
\end{lemma}

\subsection{Breaking Quadratic Barrier} \label{sec : handling dual edge in subquadratic space}
 We now show that DAG ${\mathcal D}({\mathcal A})$ and set of edges ${\mathcal A}$ stated in Theorem \ref{thm: structure for min+1} act as a dual edge sensitivity oracle for $(s,t)$-mincut and occupies subquadratic space for simple graphs. 
It follows from the proof of Lemma \ref{lem : dual edge sensitivity using residual graph} that, given residual graph $G^f$, the main objective is to design graph ${\mathcal D}_{PQ}(G\setminus \{e_1\})$. 
In order to obtain $\mathcal{D}_{PQ}(G\setminus \{e_1\})$, the residual graph $G^f$ plays a crucial role in computing a maximum $(s,t)$-flow in graph $G\setminus \{e_1\}$.
However, the \textbf{main challenge} arises in 
computing a maximum $(s,t)$-flow in graph $G\setminus\{e_1\}$ using only DAG ${\mathcal D}({\mathcal A})$ and set of edges ${\mathcal A}$ from Theorem \ref{thm: structure for min+1}. 
This is because of the following reason. 
${\mathcal D}({\mathcal A})$ initially occupies ${\mathcal O}(\min\{m,n\sqrt{\lambda}\})$ space. But, it seems quite possible that after the removal of edge $e_1$, we need the reachability information among the vertices that are mapped to single nodes in ${\mathcal D}({\mathcal A})$. This would blow up the space, as well as time, to arrive at ${\mathcal D}_{PQ}(G\setminus \{e_1\})$; 
and hence, it defeats our objective of designing a subquadratic space dual edge sensitivity oracle. 
Interestingly, we show, using  Lemma \ref{lem : mapping in Df} and structural properties of ${\mathcal D}({\mathcal A})$, that such event never occurs; and ${\mathcal D}({\mathcal A})\cup {\mathcal A}$ is sufficient to design ${\mathcal D}_{PQ}(G\setminus \{e_1\})$ using ${\mathcal O}(\min\{m,n\sqrt{\lambda}\})$ space and ${\mathcal O}(\min\{m,n\sqrt{\lambda}\})$ time.

An edge $e$ is said to be mapped to a node $\mu$ in ${\mathcal D}({\mathcal A})$ if both endpoints of $e$ are mapped to $\mu$. It follows from Lemma \ref{lem : mapping in Df} and Lemma \ref{lem : every edge has flow in DF} that if both failed edges are mapped to nodes of ${\mathcal D}({\mathcal A})$ or belong to set ${\mathcal A}$, then the capacity of $(s,t)$-mincut remains unchanged.
Therefore, without loss of generality, we assume that endpoints of edge $e_1$ are mapped to different nodes of ${\mathcal D}({\mathcal A})$.

\subsubsection*{Construction of ${\mathcal D}_{PQ}(G\setminus \{e_1\})$ using Structure ${\mathcal D}({\mathcal A})$ and Edge Set ${\mathcal A}$}
Let us determine whether failure of $e_1$ reduces the capacity of $(s,t)$-mincut in $G$. 
It follows from Lemma \ref{lem : every edge has flow in DF} that $e_1$ must carry flow. So, there exists a $(t,s)$-path $P_r$ in the residual graph containing the edge $(y_1,x_1)$ such that for every edge in $P_r$, the corresponding edge in $G$ is carrying flow.
Interestingly, the following lemma 
ensures that there is also a path $P$ in ${\mathcal D}({\mathcal A})\cup {\mathcal A}$ such that $P$ is a quotient path of $P_r$.  

\begin{lemma} \label{lem : mapping of paths in Df}
    For any pair of vertices $u,v$, let $\mu$ and $\nu$ be the nodes to which $u$ and $v$ are mapped in $\mathcal{D}(\mathcal{A})$. There exists an $(u,v)$-path $Q_1$ in $G^f$ if and only if there exists a $(\mu,\nu)$-path $Q_2$ in ${\mathcal D}({\mathcal A})\cup {\mathcal A}$. Moreover, $Q_2$ is a quotient path of $Q_1$.
\end{lemma}
\begin{proof}
    The proof of the forward direction is immediate since ${\mathcal D}({\mathcal A})\cup {\mathcal A}$ is a quotient graph of residual graph $G^f$. Let us prove the converse part. Suppose there is an $(\mu,\nu)$-path in ${\mathcal D}({\mathcal A})\cup {\mathcal A}$. Since $G$ is undirected, by construction, every node of ${\mathcal D}({\mathcal A})$ corresponds to an SCC in $G^f$. Moreover, since ${\mathcal D}({\mathcal A})\cup {\mathcal A}$ is a quotient graph of residual graph $G^f$, there is a path in $G^f$ from every vertex that is mapped to $\mu$ to every vertex that is mapped to $\nu$. 
\end{proof}
Now, to obtain the $(s,t)$-mincut capacity in $G\setminus \{e_1\}$, we need to first reduce the value of $(s,t)$-flow in $G$ by $1$ using the path $P_r$. As discussed in Appendix \ref{sec: dual edge using residual graph}, this can be achieved by performing
\textsc{Update\_Path}$(G^f,P_r)$, and then, removing the pair of edges $(x_1,y_1)$ and $(y_1,x_1)$ from the resulting graph corresponding to the failed edge $e_1$ in $G$.
 Let $f_1$ be the obtained $(s,t)$-flow in graph $G\setminus \{e_1\}$ and let $G^{f_1}$ be the corresponding residual graph after this operation. 
It follows immediately from Lemma \ref{lem : mapping of paths in Df} that we can achieve the same objective by performing
\textsc{Update\_Path}$({\mathcal D}({\mathcal A})\cup {\mathcal A},P)$, and then, removing the pair of edges $(x_1,y_1)$ and $(y_1,x_1)$ from the resulting graph corresponding to the failed edge $e_1$ in $G$.
Let $D_1$ be the resulting graph from ${\mathcal D}({\mathcal A})\cup {\mathcal A}$. 
The following lemma immediately follows from the construction of $D_1$ and Lemma \ref{lem : mapping of paths in Df}; and it states the relationship between graph $D_1$ and $G^{f_1}$.

\begin{lemma} \label{lem : d1 graph property}
    Graph $D_1$ satisfies the following two properties. 
    \begin{enumerate}
        \item $D_1$ is a quotient graph of $G^{f_1}$, where $f_1$ is the $(s,t)$-flow in $G\setminus \{e_1\}$ after applying operation \textsc{Update\_Path}$(G^f,P_r)$ and removal of edges $(x_1,y_1)$, $(y_1,x_1)$ from the resulting graph.
        \item A pair of vertices $u,v$ are mapped to the same node in ${\mathcal D}({\mathcal A})$ if and only if $u,v$ are mapped to the same node in $D_1$.
    \end{enumerate}
\end{lemma}
We have now graph $G\setminus \{e_1\}$ and $(s,t)$-flow $f_1$ of value $\lambda-1$. The aim is now to determine using $D_1$ whether $(s,t)$-flow $f_1$ is a maximum $(s,t)$-flow in $G\setminus \{e_1\}$ or the value of $(s,t)$-flow can be increased by $1$. 
For this purpose, we verify whether an $(s,t)$-path $P'$ exists in $D_1$.  Lemma \ref{lem : d1 graph property}(2) ensures that the existence of path $P'$ can be determined in ${\mathcal O}(\min\{m,n\sqrt{\lambda}\})$ time. Suppose path $P'$ exists. By Lemma \ref{lem : d1 graph property}(1), the capacity of $(s,t)$-mincut is $\lambda$ in $G\setminus \{e_1\}$. 
It follows from Lemma \ref{lem : d1 graph property}(1), and the construction of $D_1$ and $G^{f_1}$, that the mapping stated in Lemma \ref{lem : mapping of paths in Df} remains intact between $D_1$ and $G^{f_1}$.
Hence, there is an $(s,t)$-path $P_r'$ in $G^{f_1}$ such that $P'$ is a quotient path of $P_r'$. Upon updating path $P_r'$ in $G^{f_1}$ by following the construction of residual graph 
(refer to Section \ref{sec : preliminaries}), let $f_2$ be the maximum $(s,t)$-flow of value $\lambda$ obtained in $G\setminus \{e_1\}$.
Let $G^{f_2}$ denote the corresponding residual graph. However, by Lemma \ref{lem : d1 graph property}, we update path $P'$ in $D_1$, instead of updating path $P_r'$ in $G^{f_1}$, by following the construction of residual graph. Let $D_2$ be the resulting graph.  
If $P'$ does not exist, then the capacity of $(s,t)$-mincut is $\lambda-1$ in $G\setminus \{e_1\}$, $D_1=D_2$, $f_1=f_2$, and $G^{f_1}=G^{f_2}$. Now, the following lemma states the properties of graph $D_2$. 
\begin{lemma} \label{lem : property of D2}
    Graph $D_2$ satisfies the following three properties. 
    \begin{enumerate}
        \item $D_2$ is a quotient graph of $G^{f_2}$, where $f_2$ is a maximum $(s,t)$-flow in $G\setminus \{e_1\}$.
        \item A pair of vertices $u,v$ are mapped to the same node in ${\mathcal D}({\mathcal A})$ if and only if $u,v$ are mapped to the same node in $D_2$.
         \item If a pair of vertices $u,v$ is mapped to the same node in $D_2$ then $u,v$ belong to the same SCC in $G^{f_2}$.
    \end{enumerate}
\end{lemma}
\begin{proof}
    The proof of properties (1) 
    and (2) follows directly from the construction of $D_2$ and Lemma \ref{lem : d1 graph property}. 
    We now establish the property (3) of $D_2$. It follows from Lemma \ref{lem : mapping in Df} that any $(s,t)$-cut that separates any pair of vertices mapped to the same node in ${\mathcal D}({\mathcal A})$ has capacity at least $\lambda+2$ in $G$. Removal of edge $e_1$ can reduce the capacity of any $(s,t)$-cut by at most $1$. Therefore, by property (2), the least capacity $(s,t)$-cut in $G\setminus \{e_1\}$ that separates any pair of vertices $u,v$ mapped to the same node in $D_2$ has capacity at least $\lambda+1$. Moreover, the capacity of $(s,t)$-mincut in $G\setminus \{e_1\}$ is at most $\lambda$. Since graph $G$ is undirected, by using the proof of Lemma \ref{lem :mapping of nodes in Dpq}, vertices $u,v$ must belong to the same SCC in $G^{f_2}$.
\end{proof}
We now obtain DAG ${\mathcal D}_{PQ}(G\setminus \{e_1\})$ from graph $D_2$. Let $\mathbb{D}$ be the graph obtained by contracting every SCC of $D_2$ into a single node. The following lemma establishes a relation between ${\mathcal D}_{PQ}(G\setminus \{e_1\})$ and $\mathbb{D}$.
\begin{lemma} \label{lem : mapping of node in Dpq and D}
     A pair of vertices $u,v$ are mapped to the same node in $\mathbb{D}$ if and only if $u,v$ are mapped to the same node in ${\mathcal D}_{PQ}(G\setminus \{e_1\})$.  
\end{lemma}
\begin{proof}
    Suppose there is a pair of vertices $u,v$ that are mapped to the same node in $\mathbb{D}$. It follows from the construction of $\mathbb{D}$ that $u,v$ are either mapped to the same node in $D_2$ or belong to an SCC in $D_2$. In the former case, by Lemma \ref{lem : property of D2}(3), $u,v$ belong to the same SCC in $G^{f_2}$. We now consider the latter case. It follows from Lemma \ref{lem : property of D2}(1) and (3) that a pair of vertices belonging to an SCC in $D_2$ also belongs to the same SCC in $G^{f_2}$. Hence, by construction of ${\mathcal D}_{PQ}(G\setminus \{e_1\})$, $u,v$ are mapped to the same node in ${\mathcal D}_{PQ}(G\setminus \{e_1\})$.  

    Lets us prove the converse part. Suppose $u$ and $v$ are mapped to the same node of ${\mathcal D}_{PQ}(G\setminus \{e_1\})$. Since $G$ is undirected, it implies from the construction of ${\mathcal D}_{PQ}(G\setminus \{e_1\})$ that $u$ and $v$ belong to the same SCC in the residual graph $G^{f_2}$ for maximum $(s,t)$-flow $f_2$ in $G\setminus \{e_1\}$. Therefore, by Lemma \ref{lem : property of D2}(1), $u$ and $v$ either belong to the same node in $D_2$ or in an SCC in $D_2$. Hence, by construction of ${\mathbb D}$, $u$ and $v$ are mapped to the same node in $\mathbb{D}$. 
\end{proof}
It follows from Lemma \ref{lem : property of D2}(1), Lemma \ref{lem : mapping of node in Dpq and D}, and the construction of ${\mathcal D}_{PQ}(G\setminus \{e_1\})$ that the set of edges of $\mathbb{D}$ is the same as the set of edges of ${\mathcal D}_{PQ}(G\setminus \{e_1\})$. Therefore, the following lemma is immediate.
\begin{lemma} \label{lem : D is same as Dpq}
    Graph $\mathbb{D}$ is the DAG ${\mathcal D}_{PQ}(G\setminus \{e_1\})$.
\end{lemma}
It follows from Lemma \ref{lem : property of D2}(2) that $D_2$ occupies ${\mathcal O}(\min\{{m,n\sqrt{\lambda}}\})$ space, and hence, the contraction of each SCC in $D_2$ takes ${\mathcal O}(\min\{{m,n\sqrt{\lambda}}\})$ time. Therefore by Lemma \ref{lem : D is same as Dpq}, given ${\mathcal D}({\mathcal A})$ and set of edges ${\mathcal A}$, there is an algorithm that, 
using ${\mathcal O}(\min\{{m,n\sqrt{\lambda}}\})$ space, can compute graph ${\mathcal D}_{PQ}(G\setminus \{e_1\})$ in ${\mathcal O}(\min\{{m,n\sqrt{\lambda}}\})$ time. Now, similar to the discussion in Appendix \ref{sec: dual edge using residual graph} and Lemma \ref{lem : single edge failure}, it is easy to observe that, given ${\mathcal D}_{PQ}(G\setminus \{e_1\})$, only ${\mathcal O}(\min\{{m,n\sqrt{\lambda}}\})$ time  is required to report an $(s,t)$-mincut $C$ and its capacity after the failure of edges $e_1,e_2$ in graph $G$. 
To report the set of edges contributing to $(s,t)$-mincut $C$, we report every edge of ${\mathcal D}_{PQ}(G\setminus \{e_1\})$ whose one endpoint is in $C$ and the other is in $\overline{C}$. This can also be accomplished in the same ${\mathcal O}(\min\{{m,n\sqrt{\lambda}}\})$ time.
This, along with Theorem \ref{thm : insertion for multi-graphs}, completes the proof of the following theorem. 
\begin{theorem} [Dual Edge Sensitivity Oracle for Simple Graphs] \label{thm : sensitivity oracle for multi-graphs}
    Let $G$ be an undirected simple graph on $n$ vertices and $m$ edges with designated source and sink vertices $s$ and $t$, respectively. Let $\lambda$ be the capacity of $(s,t)$-mincut in $G$. There exists a data structure occupying ${\mathcal O}(\min \{m,n\sqrt{\lambda}\})$ space that can report an $(s,t)$-mincut in ${\mathcal O}(\min \{m,n\sqrt{\lambda}\})$ time after the failure or insertion of any pair of given query edges in $G$.
\end{theorem}
For simple graphs, the capacity of $(s,t)$-mincut is bounded by $n-1$. Therefore,  Theorem \ref{thm : sensitivity oracle for multi-graphs} leads to Theorem \ref{thm : sensitivity oracle for simple graphs}. 

We now analyze the time taken for computing our compact structure stated in Theorem \ref{thm: structure for min+1}. It follows from Theorem \ref{thm : anchor edges} that, given maximum $(s,t)$-flow $f$, we can compute set ${\mathcal A}$ in ${\mathcal O}(m)$ time. Moreover, given ${\mathcal A}$ and $f$, the construction of ${\mathcal D}({\mathcal A})$ requires ${\mathcal O}(m)$ time. This completes the proof of Theorem \ref{thm : preprocessing of data srtcuture D}.

\newpage
\section{Handling Dual Edge Insertions} \label{sec: handling dual edge insertion}
%
Suppose the two query edges to be inserted into simple graph $G$ are 
$e_1=(x_1,y_1)$ and $e_2=(x_2,y_2)$. 
For any directed graph $H_d$ and an undirected edge $e=(a,b)$, let $H_d\cup \{e\}$ denote the graph obtained from $H_d$ after insertion of pair of directed edges $(a,b)$ and $(b,a)$. We first state the following immediate corollary of Lemma \ref{lem : ford fulkerson augmenting paths}.
\begin{lemma} [\cite{ford_fulkerson_1956}] \label{lem : insertion in Gf}
    Upon insertion of an edge $e$ in $G$, the value of maximum $(s,t)$-flow can increase in $G\cup \{e\}$ if and only if there is an $(s,t)$-path in $(G^f \cup \{e\})$. 
\end{lemma}

The following lemma shows that $\mathcal{D}_{PQ}(G)$ can be used to report an $(s,t)$-mincut after the insertion of a single edge in $\mathcal{O}(n)$ time.
\begin{lemma} [\cite{DBLP:journals/mp/PicardQ80} and Lemma 4.4 in \cite{baswana2023minimum+}] \label{lem : single edge insertion}
     For any undirected multi-graph $\mathcal{G}$, upon the insertion of any edge $e=(u,v)$, the capacity of $(s,t)$-mincut increases by $1$ if and only if $u$ is mapped to sink $\mathbb{S}$ and $v$ is mapped to source $\mathbb{T}$ or vice-versa in ${\mathcal D}_{PQ}(\mathcal{G})$. Moreover, 
     the set of vertices mapped to $\mathbb{S}$ defines an $(s,t)$-mincut in $\mathcal{G}\cup \{e\}$.
\end{lemma}
Similar to handling the failure of edges in Appendix \ref{sec : handling dual edge in subquadratic space}, our main objective is to construct $\mathcal{D}_{PQ}(G\cup \{e_1\})$ by only using DAG $\mathcal{D}(\mathcal{A})$ and set of edges $\mathcal{A}$ from Theorem \ref{thm: structure for min+1}.
Using the fact $\mathcal{D}(\mathcal{A})\cup \mathcal{A}$ is a quotient graph of $G^f$, along a similar lines to the proof of Lemma \ref{lem : mapping of paths in Df}, we can establish the followin lemma.

\begin{lemma} \label{lem : path in Gf and Df for insertion}
    Upon the insertion of any edge $e$ in $G^f$, there exists an $(s,t)$-path $P_r$ in $G^f\cup \{e\}$ if and only if there exists an $(s,t)$-path $P$ in  $(\mathcal{D}(\mathcal{A})\cup \mathcal{A})\cup \{e\}$. Moreover, $P$ is a quotient path of $P_r$.
    \label{lem : equivalence between D(F) and Gf}
\end{lemma}

We now use $(\mathcal{D}(\mathcal{A})\cup \mathcal{A}) \cup \{e_1\}$ to determine whether flow $f$ is also a maximum $(s,t)$-flow in $G \cup \{e_1\}$ or whether its value can increase by $1$.
We verify whether there exists an $(s,t)$-path $P$ in graph $(\mathcal{D}(\mathcal{A})\cup \mathcal{A}) \cup \{e_1\}$. Suppose path $P$ exists.
It follows from Lemma \ref{lem : path in Gf and Df for insertion} that there is an $(s,t)$-path $P_r$ in $G^f$ such that $P$ is a quotient path of $P_r$.
So, by Lemma \ref{lem : insertion in Gf}, the capacity of $(s,t)$-mincut is $\lambda+1$ in $G\cup \{e_1\}$.
Upon updating path $P_r$ in $G^{f}$ by following the construction of residual graph, let $f_2$ be the maximum $(s,t)$-flow of value $\lambda+1$ obtained in $G\cup\{ e_1\}$ and let $G^{f_2}$ be the corresponding residual graph.
Since $\mathcal{D}(\mathcal{A})\cup \mathcal{A}$ is a quotient graph of $G^f$, instead of updating $P_r$, we update path $P$ in $(\mathcal{D}(\mathcal{A})\cup \mathcal{A}) \cup \{e_1\}$ by following the construction of residual graph (refer to Section \ref{sec : preliminaries}) to obtain graph $D_2$. 
The above steps can be executed in ${\mathcal O}(\min\{m,n\sqrt{\lambda}\})$ time.
If there exists no $(s,t)$-path in $(\mathcal{D}(\mathcal{A})\cup \mathcal{A}) \cup \{e_1\}$, then $D_2$ is the same as $(\mathcal{D}(\mathcal{A})\cup \mathcal{A}) \cup \{e_1\}$, $G^{f_2}=G^f$, and the capacity of $(s,t)$-mincut is $\lambda$ in $G\cup \{e_1\}$. We now establish the following relation between graph $D_2$ and $G^{f_2}$.

\begin{lemma} \label{lem : property of D2 for insertion}
    Graph $D_2$ satisfies the following three properties. 
    \begin{enumerate}
        \item $D_2$ is a quotient graph of $G^{f_2}$, where $f_2$ is a maximum $(s,t)$-flow in $G\cup \{e_1\}$.
        \item A pair of vertices $u,v$ are mapped to the same node in ${\mathcal D}({\mathcal A})$ if and only if $u,v$ are mapped to the same node in $D_2$.
         \item If a pair of vertices $u,v$ is mapped to a node in $D_2$ then $u,v$ belong to the same SCC in $G^{f_2}$.
    \end{enumerate}
\end{lemma}
\begin{proof}
    The proof of properties (1) and (2) follows directly from the construction of $D_2$. We now establish the property (3) of $D_2$. 
    Suppose $u,v$ are a pair of vertices that are mapped to the same node in $D_2$. It follows from Theorem \ref{thm: structure for min+1} that any $(s,t)$-cut that separates $u,v$ has capacity at least $\lambda+2$ in $G$. Insertion of edge $e_1$ can only increase the capacity of any $(s,t)$-cut. Therefore, by property (2), the least capacity $(s,t)$-cut in $G\cup \{e_1\}$ that separates vertices $u,v$ has capacity at least $\lambda+2$. Moreover, the capacity of $(s,t)$-mincut in $G\cup \{e_1\}$ is at most $\lambda+1$. Therefore, by Lemma \ref{lem :mapping of nodes in Dpq}, $u,v$ are mapped to the same node in $\mathcal{D}_{PQ}(G\cup \{e_1\})$. By construction of $\mathcal{D}_{PQ}(G\cup \{e_1\})$, $u,v$ belong to the same SCC in $G^{f_2}$.
 \end{proof} 
We now obtain DAG ${\mathcal D}_{PQ}(G\cup \{e_1\})$ from graph $D_2$. Let ${\mathcal I}$ be the graph obtained by contracting every SCC of $D_2$. The following relation between ${\mathcal D}_{PQ}(G\cup \{e_1\})$ and ${\mathcal I}$ can be established along similar lines to Lemma \ref{lem : mapping of node in Dpq and D} using Lemma \ref{lem : property of D2 for insertion}.

\begin{lemma} \label{lem : mapping of node in Dpq and I}
     A pair of vertices $u,v$ are mapped to the same node in ${\mathcal I}$ if and only if $u,v$ are mapped to the same node in ${\mathcal D}_{PQ}(G\cup \{e_1\})$.  
\end{lemma}
    
It follows from Lemma \ref{lem : property of D2 for insertion}(1) and Lemma \ref{lem : mapping of node in Dpq and I} that the set of edges in $\mathcal{I}$ is the same as the set of edges in ${\mathcal D}_{PQ}(G\cup \{e_1\})$. Therefore, the following lemma is immediate.

\begin{lemma} \label{lem : Dpq same as D2 for insertion}
    Graph ${\mathcal I}$ is the graph $\mathcal{D}_{PQ}(G\cup \{e_1\})$.
\end{lemma}

It follows from Lemma \ref{lem : Dpq same as D2 for insertion} that given ${\mathcal D}({\mathcal A})$ and set of edges ${\mathcal A}$, there is an algorithm that, 
using ${\mathcal O}(\min\{{m,n\sqrt{\lambda}}\})$ space, can compute graph ${\mathcal D}_{PQ}(G\cup \{e_1\})$ in ${\mathcal O}(\min\{{m,n\sqrt{\lambda}}\})$ time.
Hence, we can use Lemma \ref{lem : single edge insertion} to check whether the insertion of edge $e_2$ in $G \cup \{e_1\}$ increases the capacity of $(s,t)$-mincut.
Suppose $e_2$ satisfies the condition stated in Lemma \ref{lem : single edge insertion} for $\mathcal{D}_{PQ}(G\cup \{e_1\})$. If the $(s,t)$-mincut capacity in $G\cup \{e_1\}$ is $\lambda+1$ (likewise $\lambda$), then after the insertion of two edges $e_1,e_2$, the capacity of $(s,t)$-mincut in $G$ is $\lambda+2$ (likewise $\lambda+1$). 
If $e_2$ fails to satisfy the condition stated in Lemma \ref{lem : single edge insertion} for ${\mathcal D}_{PQ}(G\cup \{e_1\})$, the capacity of $(s,t)$-mincut in $G$ after the insertion of edges $e_1,e_2$ is the same as the capacity of $(s,t)$-mincut in $G\cup \{e_1\}$.
Moreover, the reported $(s,t)$-mincut using Lemma \ref{lem : single edge insertion} in $\mathcal{D}_{PQ}(G\cup \{e_1\})$ is also an $(s,t)$-mincut in $G$ after the insertion of two edges $e_1,e_2$.
Therefore, overall ${\mathcal O}(\min\{{m,n\sqrt{\lambda}}\})$ time  is required to report an $(s,t)$-mincut and its capacity after the insertion of edges $e_1,e_2$ in graph $G$. This leads to the following theorem. 
\begin{theorem} [Dual Edge Insertions for Simple Graphs] \label{thm : insertion for multi-graphs}
    Let $G$ be an simple graph on $n$ vertices and $m$ edges with designated source and sink vertices $s$ and $t$ respectively. There exists a data structure occupying ${\mathcal O}(\min \{m,n\sqrt{\lambda}\})$ space that can report an $(s,t)$-mincut in ${\mathcal O}(\min \{m,n\sqrt{\lambda}\})$ time after the insertion of any pair of given query edges in $G$.
\end{theorem}

\section{Computing Minimum+1 (s,t)-cut in Directed Multi-graphs} \label{sec : minimum+1 algorithm in multigraphs}

For graphs with integer edge capacities, Gabow \cite{gabow1991matroid} designed an algorithm that takes $\tilde{{\mathcal O}}(m\lambda)$ time to compute a global mincut. By using this algorithm of \cite{gabow1991matroid} and our result in Theorem \ref{thm : equivalence between second st mincut and global}(1), 
we immediately arrive at the algorithm stated in the following theorem. 
\begin{theorem}  \label{thm : second mincut using all pair mincut}
    For any directed multi-graph on $n$ vertices and $m$ edges with designated source and sink vertices $s$ and $t$ respectively, there is an algorithm that, given any maximum $(s,t)$-flow, can compute a second $(s,t)$-mincut of capacity $\lambda+\kappa$, where $\kappa\ge 0$ is an integer, in $\tilde{{\mathcal O}}(m\kappa)$ time. 
\end{theorem}
Let us consider graph $G$ to be a directed multi-graph for this section. Suppose there is a $(\lambda+1)$ $(s,t)$-cut in $G$. It follows that the capacity of second $(s,t)$-mincut is $\lambda+1$ in $G$. 
By Theorem \ref{thm : second mincut using all pair mincut} with $\kappa=1$, given any maximum $(s,t)$-flow, 
we can compute a $(\lambda+1)$ $(s,t)$-cut in $G$ in $\tilde{{\mathcal O}}(m)$ time. 
We now present an algorithm that given any maximum $(s,t)$-flow in $G$, can compute a $(\lambda+1)$ $(s,t)$-cut, if it exists, in only ${\mathcal O}(m)$ time.
Our algorithm uses a different, as well as simpler, approach compared to the algorithm in Theorem \ref{thm : second mincut using all pair mincut} that uses the global mincut computation of \cite{gabow1991matroid} in directed graphs. Moreover, the insights established to arrive at our algorithm pave the way to compute all the anchor edges in undirected multi-graphs in ${\mathcal O}(m)$ time (refer to Section \ref{sec : anchor edge computation}). 

Along a similar line to the design of Algorithm \ref{alg : second mincut using global mincut} stated in Theorem \ref{thm : second minimum (s,t)-cut}, given a second $(s,t)$-mincut in a graph with at most two $(s,t)$-mincuts $\{s\}$ and $V\setminus \{t\}$, we can design an algorithm that can compute a second $(s,t)$-mincut in ${\mathcal O}(m)$ time for graph $G$ (refer to Appendix \ref{sec : extension to egenral graphs}).
So, for the rest of the section, we assume that graph $G$ has at most two $(s,t)$-mincuts $\{s\}$ and $V\setminus\{t\}$. The following concept of $(A,B)$-mincut for sets $A,B\subset V$ is defined for better exposition of our result.
 \begin{definition}[($A,B$)-mincut]
     For any pair of sets $\emptyset\subset A,B\subset V$ with $A\cap B=\emptyset$, a cut $C$ is said to be an $(A,B)$-cut if $A \subseteq  C$ and $B \subseteq \overline{C}$. An $(A,B)$-mincut is an $(A,B)$-cut of the least capacity.   
 \end{definition}

\subsection{Graphs with Exactly One (s,t)-mincut} \label{sec : min+1 in graphs with exactly 2 st mincuts}
Suppose $G$ has exactly one $(s,t)$-mincut $V\setminus\{t\}$. The algorithm is similar for graph $G$ with $(s,t)$-mincut only $\{s\}$. 
For computing a $(\lambda+1)$ $(s,t)$-cut, similar to the algorithm stated in Lemma \ref{lem : n maxflow}, we aim to efficiently find a vertex $x \in V\setminus \{s,t\}$, if exists, such that there is a $(\lambda+1)$ $(s,t)$-cut $C$ with $x\in \overline{C}$. We now establish a necessary and sufficient condition for identifying such a vertex $x$.

\begin{lemma} \label{lem : edge disjoint paths and min+1}
     For every vertex $x \in V\setminus \{s,t\}$, there is a $(\lambda+1)$ $(s,t)$-cut $C$ with $x \in \overline{C}$ in $G$ if and only if there is exactly one edge-disjoint path from $s$ to $x$ in $G^f$. 
\end{lemma}
\begin{proof}
    Let the capacity of $(s,x)$-mincut be $\lambda_1$ and $(s,\{x,t\})$-mincut be $\lambda_2$. It follows that $\lambda_1\le \lambda_2$. Let $C$ be an $(s,\{x,t\})$-mincut. Since $t$ has indegree $0$, $C\setminus \{t\}$ is an $(s,x)$-cut of capacity $\lambda_2$, and hence, $\lambda_2\le \lambda_1$. It follows that $\lambda_1=\lambda_2$. 

    If there is a $(\lambda+1)$ $(s,t)$-cut $C$ with $x \in \overline{C}$ in $G$, by Theorem \ref{thm : min+k in residual graph}, $\lambda_2\le 1$.    
    Since $V\setminus\{t\}$ is the only $(s,t)$-mincut in $G^f$, by Lemma \ref{lem : alternative maxflow mincut}, $\lambda_2=1$. This implies $\lambda_1=1$ and, by Menger's theorem \cite{menger1927}, the number of edge-disjoint paths from $s$ to $x$ is one.

    If there is exactly one edge-disjoint path from $s$ to $x$ in $G^f$, by Menger's Theorem\cite{menger1927}, $\lambda_1=1$. Hence we get $\lambda_2=1$, and  by Theorem \ref{thm : min+k in residual graph}, there exists a $(\lambda+1)$ $(s,t)$-cut $C$ with $x \in \overline{C}$ in $G$.
\end{proof}
It follows from Lemma \ref{lem : edge disjoint paths and min+1} that we only need to efficiently find a vertex $x$ such that there is exactly one  edge disjoint path from $s$ to $x$ in $G^f$.
We use the concept of \textit{dominance} with respect to vertex $s$ to achieve this goal as follows.
A vertex $v$ is called a \textit{dominator} of a vertex $u$ if every $(s,u)$-path contains vertex $v$.
A dominator tree of a graph 
is defined as follows. 

\begin{theorem}[Dominator Tree \cite{tarzandominator1979}]\label{thm : dominator}
     Let $\mathcal{G}$ be any directed multi-graph. Let $T_{\mathcal{G}}$ denote the dominator tree of $\mathcal{G}$ on vertex set $V$ with root vertex $s$. $T_{\mathcal{G}}$ is a directed tree rooted at $s$ such that the following property holds. For any pair of vertices $u,v \in V$, $v$ is an ancestor of $u$ in $T_\mathcal{G}$ if and only if $v$ is a dominator of $u$ in $\mathcal{G}$. 
\end{theorem}
Observe that in graph $G^f$, there is exactly one edge-disjoint path from $s$ to $x$ if and only if there exists an edge $e'$ such that every $(s,x)$-path contains $e'$. However, for any vertex $u$, even if $v$ is a dominator of a vertex $u$, it is quite possible that there are more than one edge-disjoint $(s,u)$-paths in $G^f$. So, a dominator tree for $G$ does not immediately help in determining whether there is exactly one edge-disjoint $(s,u)$-path in $G^f$. We now construct a graph $\mathcal{H}$ from $G^f$ whose dominator tree helps in achieving our objective.

\paragraph*{Construction of $\mathcal{H}=(V',E')$:} 
For each edge $(u,v)$ in $G^f$, split edge $(u,v)$ into two edges $(u,w)$ and $(w,v)$ using a new vertex $w$. It follows that the number of vertices, as well as edges, in $\mathcal{H}$ is $\mathcal{O}(m)$.

To design our algorithm, we begin by classifying the vertices in $\mathcal{H}$ as follows. A vertex $v \in V'$ is said to be \textit{marked} if $v \notin V$, otherwise, it is said to be \textit{unmarked}. 
It follows from the construction of ${\mathcal H}$ that there exists a marked vertex $v$ in $\mathcal{H}$ corresponding to every edge $e'$ in $G^f$ and vice versa. Let $T_\mathcal{H}$ be the dominator tree of $\mathcal{H}$ with root vertex $s$. We now establish the following relation between $T_\mathcal{H}$ and graph $G^f$.

\begin{lemma}\label{lem : dominator in modified graph}
    Let $v$ be a marked vertex in $\mathcal{H}$ corresponding to an edge $e'$ in $G^f$. $v$ is an ancestor of an unmarked vertex $u$ in $T_\mathcal{H}$ if and only if every $(s,u)$-path contains $e'$ in $G^f$.
\end{lemma}
\begin{proof}
    Suppose marked vertex $v$ is an ancestor of an unmarked vertex $u$ in $T_\mathcal{H}$. It follows from Theorem \ref{thm : dominator} that each $(s,u)$-path in $\mathcal{H}$ contains $v$. It follows from the construction of $\mathcal{H}$ that each $(s,u)$-path contains the edge $e'$ (corresponding to $v$) in $G^f$. 
    
    Suppose $u$ is an unmarked vertex and every $(s,u)$-path contains edge $e'$ in $G^f$. 
    By construction of $\mathcal{H}$, every $(s,u)$-path contains marked vertex $v$ (corresponding to edge $e'$ in $G^f$) in $\mathcal{H}$.
    By Theorem \ref{thm : dominator}, $v$ is ancestor of $u$ in $T_\mathcal{H}$. 
\end{proof}    
We now use Lemma \ref{lem : dominator in modified graph} to establish a relation between the dominator tree $T_\mathcal{H}$ and $(\lambda+1)$ $(s,t)$-cuts in graph $G$.

\begin{lemma} \label{lem : dominator with min+1 st cut}
    Let $v$ be a marked vertex in $\mathcal{H}$ and let the edge in $G^f$ corresponding to $v$ be $e'=(w,u)$.
    $v$ is an internal vertex in $T_\mathcal{H}$ with child $u$ if and only if there exists a $(\lambda+1)$ $(s,t)$-cut $C$ in $G$ such that the only edge outgoing from $C$ in $G^f$ is edge $e'$.
\end{lemma}
\begin{proof}
    Suppose $v$ is an internal marked vertex in $T_\mathcal{H}$ with child $u$.
    By construction of $\mathcal{H}$, $u$ is an unmarked vertex, since $e'=(w,u)$ in $G^f$.
    By Lemma \ref{lem : dominator in modified graph}, every $(s,u)$-path contains edge $(w,u)$ in $G^f$. 
    Therefore, it is easy to show that the least capacity $(s,t)$-cut in $G^f$ that separates $\{s,w\}$ from $\{u,t\}$ has capacity one, since $t$ has indegree $0$.
    By Theorem \ref{thm : min+k in residual graph}, $C$ is a $(\lambda+1)$ $(s,t)$-cut in $G$ with $u \in \overline{C}$ such that only $e'$ contributes to $C$ in $G^f$.
    
    Suppose there exists a $(\lambda+1)$ $(s,t)$-cut $C$ in $G$ such that the only edge outgoing from $C$ in $G^f$ is edge $e'$. 
    Since $u \in \overline{C}$, it follows that every $(s,u)$-path must contain edge $e'$ in $G^f$. By Lemma \ref{lem : dominator in modified graph}, $v$ is an ancestor of $u$ (unmarked vertex) in $T_\mathcal{H}$.     
    By construction of $\mathcal{H}$, $(v,u)$ is an edge in $\mathcal{H}$. Hence, $v$ is the parent of $u$ in $T_\mathcal{H}$. 
\end{proof}

Now, we state the algorithm for computing a $(\lambda+1)$ $(s,t)$-cut in $G$.
\paragraph*{Algorithm:}
The algorithm begins by computing $G^f$ and graph $\mathcal{H}$ using maximum $(s,t)$-flow $f$ in $G$.
Compute the dominator tree $T_\mathcal{H}$ from $\mathcal{H}$ using the ${\mathcal O}(m)$ time algorithm given in \cite{dominator/alstrup1999dominators}. 
To determine the existence of a $(\lambda+1)$ $(s,t)$-cut in $G$, find a marked internal vertex in $T_\mathcal{H}$.
If there does not exist any marked internal vertex in $T_\mathcal{H}$, by Lemma \ref{lem : dominator with min+1 st cut}, there exists no $(\lambda+1)$ $(s,t)$-cut in $G$ and the algorithm terminates.
Suppose there exists an internal marked vertex $v$ in $T_\mathcal{H}$ with child $u$. 
By Lemma \ref{lem : dominator with min+1 st cut}, there is a $(\lambda+1)$ $(s,t)$-cut $C$ in $G$ such that $u \in \overline{C}$. Obtain graph $G_u$ from graph $G^f$ by adding a pair of edges between $u$ and $t$. Compute a maximum $(s,t)$-flow in graph $G_u$ using the algorithm of Ford and Fulkerson \cite{ford_fulkerson_1956}. This completes the proof of the following lemma.
\begin{lemma}\label{lem : computing min+1 in a rgaph with one st mincut}
    Let $G$ be a graph with exactly one $(s,t)$-mincut $V\setminus \{t\}$ $($or $\{s\})$. There is an algorithm that, given a maximum $(s,t)$-flow, computes a $(\lambda+1)$ $(s,t)$-cut in $G$ in $\mathcal{O}(m)$ time.
\end{lemma}

\subsection{Graphs with Exactly two (s,t)-mincuts} \label{sec : min+1 using covering}
Lemma \ref{lem : edge disjoint paths and min+1} fails to hold when $G$ has exactly $2$ $(s,t)$-mincuts: $\{s\}$ and $V \setminus\{t\}$. This is because $s$ has outdegree $0$ in the residual graph $G^f$ and there is no path from $s$ to any vertex. 
Similar to the proof of Lemma \ref{lem : n maxflow}, we can address this problem by constructing the pair of graphs $G^I$ and $G^U$ (as constructed in the algorithm for second $(s,t)$-mincut stated in Lemma \ref{lem : n maxflow}) using the covering technique of \cite{baswana2023minimum+}.
It follows that Lemma \ref{lem : computing min+1 in a rgaph with one st mincut} is satisfied in both the graphs $G^I$ and $G^U$. This leads to the following lemma. 

\begin{lemma}\label{lem : computing min+1 in a rgaph with two st mincut}
    Let $G$ be any graph with exactly two $(s,t)$-mincuts $\{s\}$ and $V\setminus \{t\}$. There is an algorithm that, given a maximum $(s,t)$-flow, computes a $(\lambda+1)$ $(s,t)$-cut in $G$ in  $\mathcal{O}(m)$ time.
\end{lemma}
Using Lemma \ref{lem : computing min+1 in a rgaph with one st mincut} and Lemma \ref{lem : computing min+1 in a rgaph with two st mincut}, we can compute a $(\lambda+1)$ $(s,t)$-cut in general graphs as shown in Appendix \ref{sec : extension to egenral graphs} with the following slight modification. Instead of working with each SCC $H$, here we compute $(\lambda+1)$ $(s,t)$-cuts in the graph $G_\mu$ constructed for the node $\mu$ of ${\mathcal D}_{PQ}(G)$ corresponding to $H$ (refer to Appendix \ref{sec : extension to egenral graphs}).
Observe that each edge $(\mu_1,\mu_2)$ in $\mathcal{D}_{PQ}(G)$ appears in exactly two graphs -- $G_{\mu_1}$ and $G_{\mu_2}$. Moreover, due to the use of Covering Technique \cite{baswana2023minimum+}, each edge and vertex in $G_\mu$ is counted at most twice. Hence, the total number of vertices and edges across all graphs $G_\mu$ are $\mathcal{O}(n)$ and $\mathcal{O}(m)$ respectively. This completes the proof of Theorem \ref{thm : minimum+1 (s,t)-cut}.

\section{Computation of All Anchor Edges in Undirected Multi-Graphs} \label{sec : anchor edge computation}
In this section, we design an efficient algorithm to compute the set of all anchor edges $\mathcal{A}$ in an undirected multi-graph for any fixed maximum $(s,t)$-flow $f$.
It follows from  Lemma \ref{lem : every edge has flow in Dpq} and Definition \ref{def : anchor edge} that no anchor edge appears in $\mathcal{D}_{PQ}(G)$ since $f(e)=0$. By Lemma \ref{lem :mapping of nodes in Dpq} and the construction of $\mathcal{D}_{PQ}(G)$, the following lemma is immediate.
\begin{lemma}\label{lem : anchor edges and residual graph}
    If $(u,v) \in \mathcal{A}$ then $u$ and $v$ are mapped to the same SCC in $G^f.$
\end{lemma}
It follows from Lemma \ref{lem : anchor edges and residual graph} that set of edges $\mathcal{A}$ can be partitioned based on the nodes to which they are mapped. Hence, we now provide the construction of a graph $G_\mu^U$ to compute all the edges belonging to $\mathcal{A}$ whose endpoints are mapped to node $\mu$ in $\mathcal{D}_{PQ}(G)$. The construction of $G_\mu^U$ is the same as $G_\mu$ but the contraction of vertices is done in $G$ instead of $G^f$ by exploiting the topological ordering of DAG $\mathcal{D}_{PQ}(G)$. 

\paragraph*{Construction of $G_\mu$:}Let $\tau$ be a topological ordering of nodes in $\mathcal{D}_{PQ}(G)$ that begins with ${\mathbb T}$ and ends with ${\mathbb S}$. Graph $G_{\mu}^U$ is obtained by modifying graph $G$ as follows. The set of vertices mapped to the nodes that precede $\mu$ in $\tau$ is contracted into a sink vertex $t'$. Similarly, the set of vertices mapped to the nodes that succeed $\mu$ in $\tau$ is contracted into source vertex $s'$. If $\mu=\mathbb{S}$ (likewise $\mathbb{T}$),
then we map exactly $s$ to $s'$ (likewise $t$ to $t'$).

Without causing ambiguity, we denote an $(s',t')$-cut in $G_\mu^U$ by an $(s,t)$-cut.
By Lemma \ref{lem : anchor edges and residual graph} and the construction of $G_\mu^U$, the following lemma is immediate.
\begin{lemma}\label{lem : at most 2 st mincuts and ancor edges preserved}
    In graph $G_\mu^U$, the following assertions hold.
    \begin{enumerate}
        \item The capacity of $(s,t)$-mincut is $\lambda$.
        \item There are at most two $(s,t)$-mincuts -- $\{s\}$ and $V\setminus\{t\}$.
        \item For any edge $(u,v) \in \mathcal{A}$, $u$ and $v$ are mapped to node $\mu$ if and only if $(u,v)$ is an anchor edge in $G_\mu^U$.
    \end{enumerate} 
\end{lemma}
It is evident from Lemma \ref{lem : at most 2 st mincuts and ancor edges preserved} that we need to identify anchor edges in graphs with at most two $(s,t)$-mincuts $\{s\}$ and $V\setminus \{t\}$. Moreover, by the use of covering technique \cite{baswana2023minimum+}, we only need to identify anchor edges in graphs with exactly one $(s,t)$-mincut $\{s\}$ or $V\setminus \{t\}$. 
Consider $G$ to be an undirected multi-graph with exactly one $(s,t)$-mincut $V\setminus \{t\}$. The algorithm is similar when the only $(s,t)$-mincut is $\{s\}$.
Consider the graph $\mathcal{H}$ constructed from $G^f$ in Appendix \ref{sec : min+1 in graphs with exactly 2 st mincuts}.
We now establish a relation between the anchor edges and the internal marked vertices in the dominator tree $T_\mathcal{H}$. 
\begin{lemma} \label{lem : anchor edge with internal vertex in dominator tree}
    Let $T_\mathcal{H}$ be the dominator tree for graph $\mathcal{H}$. Suppose $v$ is a marked vertex in $\mathcal{H}$ and the edge corresponding to $v$ in $G^f$ is $e'=(w,u)$. 
    $v$ is an internal vertex with a child $u$ in $T_\mathcal{H}$ if and only if the undirected edge $(w,u)$ corresponding to $e'$ is an anchor edge in $G$ and there is a $(\lambda+1)$ $(s,t)$-cut $C$ with $w \in C, u\in \overline{C}$.
\end{lemma}
\begin{proof}
    Suppose $v$ is an internal vertex in $T_\mathcal{H}$ with child $u$. It follows from Lemma \ref{lem : dominator with min+1 st cut} that there exists a $(\lambda+1)$ $(s,t)$-cut $C$ in $G$ such that the edge $e'=(w,u)$ contributes to $C$ in $G^f$. Hence, $w \in C,u \in \overline{C}$.
     It follows from Theorem \ref{thm : maxflow and minimum+1 cut} and the construction of $G^f$ that the undirected edge $(w,u)$ in $G$ (corresponding to $e'$) carries zero flow.
    It follows from Theorem \ref{thm : maxflow and minimum+1 cut} and Lemma \ref{lem : unique anchor edge} that the edge carrying zero flow in the edge-set of $C$ is an anchor edge.
    Hence, edge $(w,u)$ is an anchor edge in $G$. 

    Suppose the edge $(w,u)$ is an anchor edge in $G$ and there is a $(\lambda+1)$ $(s,t)$-cut $C$ in $G$ with $w \in C, u\in \overline{C}$.
    By construction of $G^f$ and Theorem \ref{thm : maxflow and minimum+1 cut}, exactly $e'=(w,u)$ contributes to $C$ in $G^f$. 
    By Lemma \ref{lem : dominator with min+1 st cut}, marked vertex $v$ corresponding to edge $e'$ in $G^f$ is an internal vertex in $T_\mathcal{H}$ with child $u$. 
\end{proof}
Lemma \ref{lem : anchor edge with internal vertex in dominator tree} implies that it is sufficient to compute all the internal marked vertices in $T_\mathcal{H}$ corresponding to graph $G_\mu^U$ for each node $\mu$ in $\mathcal{D}_{PQ}(G)$. Similar to the analysis of the algorithm for computing a $(\lambda+1)$ $(s,t)$-cut (refer to Appendix \ref{sec : min+1 using covering}), it can be established that we can compute set ${\mathcal A}$ in ${\mathcal O}(m)$ time. 

\section{Space Occupied by $\mathcal{D}(G\setminus \mathcal{A})$ in Undirected Multi-graphs}\label{sec : size of Dpq}
In this section, we establish an upper bound on the space required by $\mathcal{D}_{PQ}$ for undirected multi-graphs using only existing results. 
We first state the following property for $\mathcal{D}_{PQ}$ which holds for undirected graphs.
\begin{lemma}[\cite{DBLP:journals/mp/PicardQ80,ford_fulkerson_1956}] \label{lem : every edge has flow in Dpq}
    Let $f$ be any maximum $(s,t)$-flow in any undirected graph ${G}$.
    For any edge $e=(x,y)$ in ${\mathcal D}_{PQ}({G})$, the corresponding undirected edge $e$ in ${ G}$ satisfies
    $f(e)=w(e)$ and $e$ carries flow in the direction $y$ to $x$. 
\end{lemma}
An $(s,t)$-flow is said to be acyclic if there is no directed cycle in which every edge carries flow in the direction of the cycle.
For any integral acyclic maximum $(s,t)$-flow in an undirected integer-weighted graph,
the number of edges carrying flow is bounded by the following lemma.
\begin{lemma}[\cite{shortlengthversionofmengerstheorem} and Lemma 12 in \cite{henzingericalp2023}] \label{lem : cardinality of set flow}
    Let ${G}$ be any undirected integer-weighted graph with no parallel edges and $(s,t)$-mincut capacity $\lambda$. For any integral acyclic maximum $(s,t)$-flow in ${G}$, the number of edges carrying flow in ${G}$ is $\mathcal{O}(\min \{m,n \sqrt{\lambda}\})$.
\end{lemma}
It follows from Lemma \ref{lem : every edge has flow in Dpq} that for every edge $e$ of ${\mathcal D}_{PQ}({\mathcal G})$, the corresponding undirected edge $e$ carries flow in any maximum $(s,t)$-flow in ${\mathcal G}$. As a result, they also carry flow in any integral acyclic maximum $(s,t)$-flow. Therefore, the following lemma is immediate from Lemma \ref{lem : cardinality of set flow}. 
\begin{lemma}\label{lem : size of Dpq in integer weighted}
    For any undirected integer-weighted graph ${G}$ with $(s,t)$-mincut capacity $\lambda$, the space occupied by DAG $\mathcal{D}_{PQ}({G})$ is $\mathcal{O}(\min\{m,n\sqrt{\lambda}\})$. 
\end{lemma}
We can easily transform any undirected multi-graph $G$ to an undirected integer-weighted graph $G'$ with no parallel edges as follows. For any pair $u,v$, if there are $q$ edges, $q >0$, between $u$ and $v$, then replace all these edges with a single edge $(u,v)$ of capacity $q$. 
Observe that for each cut $C$, $C$ has capacity $\lambda'$ in $G$ if and only if $C$ has capacity $\lambda'$ in $G'$. Hence, for any undirected multigraph $G$, we can use $\mathcal{D}_{PQ}(G')$ to store and characterize all $(s,t)$-mincuts in $G$. This, along with Lemma \ref{lem : size of Dpq in integer weighted}, gives us the following result.

\begin{lemma}\label{lem : size of Dpq}
    For any undirected multi-graph ${G}$ with $(s,t)$-mincut capacity $\lambda$, there is an undirected integer-weighted graph ${ G'}$ with no parallel edges such that 
    \begin{enumerate}
        \item the capacity of each cut in $G$ remains the same in $G'$ and 

        \item DAG $\mathcal{D}_{PQ}(G')$ occupies $\mathcal{O}(\min\{m,n\sqrt{\lambda}\})$ space.
    \end{enumerate}
\end{lemma}

\end{document}